\documentclass[11pt]{article}

\usepackage{geometry}
\usepackage{amsthm}
\usepackage{amsmath}
\usepackage{amsfonts}
\usepackage{graphicx}
\usepackage{hyperref}
\usepackage{color}
\usepackage{caption}
\usepackage{subcaption}
\usepackage{float}
\usepackage{soul}
\usepackage{multirow}
\usepackage{fullpage} 
\usepackage{xspace}
\usepackage{todonotes}
\usepackage[inline]{enumitem}
\usepackage{microtype,amssymb}
\usepackage{cleveref}
\usepackage{thm-restate}

\usepackage{framed}
\usepackage[ruled]{algorithm}
\usepackage[noend]{algpseudocode}

\algrenewcommand\textproc{}
\newtheorem{theorem}{Theorem}

\newtheorem{corollary}[theorem]{Corollary}
\newtheorem{definition}[theorem]{Definition}

\newtheorem{claim}[theorem]{Claim}
\newtheorem{obs}[theorem]{Observation}
\newtheorem{problem}{Problem}
\newtheorem{observation}[theorem]{Observation}

\def\df{{d_F}}
\def\RR{{\mathbb R}}
\def\NN{{\mathbb N}}
\def\ZZ{{\mathbb Z}}

\def\AAA{{\mathcal A}}

\def \eps{\varepsilon}
\def\bE{{\mathbb E}}

\def \fspann{{\mathrm{fspan_F}}}
\def \spann{{\mathrm{span_F}}}
\def \closure{{\mathrm{closure}}}

\DeclareMathOperator{\cost}{cost}

\DeclareMathOperator*{\argmin}{arg\,min}

\title{Time Series Decomposition using the Fréchet Distance}

\author{
Anne Driemel\thanks{University of Bonn, Germany.}
\and
Jan Höckendorff\thanks{University of Cologne, Germany. Funded by the Deutsche Forschungsgemeinschaft (DFG, German Research Foundation) – Project Number 459420781.}
\and Ioannis Psarros \thanks{Archimedes, Athena Research Center, Greece. This work has been partially supported by project MIS 5154714 of the National Recovery and Resilience Plan Greece 2.0 funded by the European Union under the NextGenerationEU Program.}
\and Christian Sohler \thanks{University of Cologne, Germany.}
}

\date{}

\begin{document}

\maketitle
\begin{abstract}
In this paper, we introduce a new data analysis problem that aims to decompose a set of univariate time series into a small set of $k$ base curves of length at most $l$ such that the sum of Fréchet distances of the time series to a ``Fréchet combination'' of the base curves is minimized. Here, a Fréchet combination allows to combine individually scaled base curves using a $k$-dimensional traversal.
We call the problem of finding a set of optimal base curves the Fréchet decomposition problem and we consider two variants: 
(a) the base curves can be arbitrary curves of bounded length and (b) the curves come from a given finite set of candidate curves. We think of the Fréchet decomposition problem as a Fréchet variant of principal component analysis.
For the case of a single base curve we develop a $(1+\varepsilon)$-approximation algorithm for the Fréchet decomposition problem.
Additionally we give an exact algorithm for the projection distance problem that asks to compute the distance of one given time series to a given set of $k$ base curves. This allows us to design an exact algorithm for the Fréchet decomposition problem for general $k$ when curves come from a fixed candidate set.

\end{abstract}
\section{Introduction}

A fundamental problem in exploratory data analysis and forecasting is to decompose a given signal into its (most important) components. Several methods for this problem are known, including principal component analysis (PCA; also known as Karhunen-Loewe transform or empirical orthogonal functions~\cite{lorenz1956empirical, hannachi2007empirical,wold1987principal}), complex principal component analysis \cite{horel1984complex}, independent component analysis \cite{JUTTEN19911}, functional principal component analysis etc. 
One drawback of principal component analysis is that it cannot deal effectively with individual temporal distortion that may happen for a relevant pattern across the input time series. 
One reason why this may happen is that the signal is not sampled uniformly, but one could also think of physical effects that could lead to such a behaviour. For example, in the area of oceanography, one of the main factors in sea surface height analysis~\cite{uebbing2024alternative} is a change that corresponds to summer/winter, and so depending on the location on the planet, the underlying signal may be shifted. 
Some variants of PCA like the complex principal component analysis, are able to find temporal shifts of patterns in the input \cite{horel1984complex}.  However, this comes at the expense of an intrinsic indeterminacy that complicates the interpretation of the results.
Kernel PCA~\cite{mika1998kernel} can be viewed as addressing similar needs. There are kernels for variants of the Fréchet distance~\cite{takeuchi2021frechet} and DTW~\cite{CuturiVBM07,CuturiB17}.
However, a fundamental problem with this approach is that the decomposition exists only in lifting to a high-dimensional feature space, and there is no reverse mapping from the feature space back to the input space. This is also called the preimage problem~\cite{honeine2011preimage}.
Otherwise classical techniques like spectral analysis \cite{priestley1981spectral}, wavelet analysis \cite{shumway2000time} 
are used to determine patterns with respect to a single time series, which makes them unsuitable for our setting.

In the statistics literature, so-called function registration aims to find suitable warping functions for the input time series to optimize the fit to a template or reference function as a preprocessing step to a (functional) principal component analysis~\cite{wang1997alignment,ramsay1998curve}. This is usually done for each input time series individually by solving a regression problem and is sensitive to the choice of the template function. Sometimes pairwise alignments are used; refer to~\cite{wang2016functional} for a recent survey. More recent approaches combine the fitting of alignments with a principal component analysis~\cite{kneip2008combining,lee2016combined}. However, these approaches typically lack convergence guarantees, not to mention a rigorous running time or complexity analysis.

This paper aims at developing new techniques to decompose a signal into components in a setting where the basic components that shape the signal can be stretched or compressed in time. 
The Fr\'echet distance is well-equipped to handle such distortion. Building upon the work on clustering of time series~\cite{DKS16, buchin2019hardness, CH23, BDGHKLS19, nath2020k, buchin2022coresets, BDR23, driemel2025nearlineartimeapproximationscheme}, we are aiming at developing a signal decomposition technique that uses the Fr\'echet distance.

More concretely, assume we are given a set  $X = \{x_1,\dots,x_n\}$ of time series. For example, each could be a time-ordered sequence of measurements of the same time-varying signal at a specific location. We are looking for $k$ base signals $b_1,\dots, b_k$ that can explain all $x_i$ simultaneously. In other words, we think of $x_i$ as being generated by the $k$ base signals and noise. 
One simple way to phrase this problem is to say the $x_i$ 
can be generated by a linear combination of the $b_i$, and we have Gaussian error. In this case, one would try to find $k$ base curves $b_1,\dots, b_k$ that minimize
\[
\sum_{i=1}^n \min_{(w_1,\dots, w_k)} \|x_i - \sum_{j=1}^k w_j \cdot b_j\|^2.
\]
The latter problem can be solved using principal component analysis.

In our case, even though we are measuring the same underlying process, we must also take into account the situation that the same signals appear in a different form (for example, they might be delayed or distorted).  
We therefore introduce a problem formulation that allows for distorted signals, based on the (discrete) Fr\'echet distance. The discrete Fr\'echet distance between two signals $x =(x_1,\dots, x_n)$ and $y = (y_1,\dots, y_m)$ is defined as the minimum over all aligned traversals of $x$ and $y$ (where one can either advance on $x$ or $y$ or both) of the maximum distance of the matched entries. It can be easily extended to a multidimensional traversal between an input time series $x_i$ and $k$ weighted base curves $w_i b_i$ where $w_i$ denotes the weight of the $i$-th base curve $b_i$. Let $B$ denote such a set of  $k$ base curves. We develop a notion of the span under the Fréchet distance, denoted with $\spann(B)$, which is a set that contains all time series that can be formed as the weighted sum of elements of $B$ under combined traversals. Let $d_F(x_i, \spann(B))$ denote the standard Hausdorff-extension of the distance to the set $\spann(B)$. We can think of this as a distance to a projection onto this set. Then our goal will be to find a set of $k$ base curves $B$ that minimize 
\[
\sum_{i=1}^n d_F(x_i,\spann(B)).
\]
We will refer to the problem of finding $B$ as the Fr\'echet Decomposition Problem. More details will be given in Section \ref{sec:problem}.
In Appendix~\ref{sec:experiments} we validate the problem definition with an small experiment on real data. The problem definition arises from interdisciplinary work in collaboration with geoscientists~\cite{uebbing2024alternative}.

\subsection{Related work}

There is extensive literature on the complexity of computing the Fréchet distance. Computing the distance for two curves of complexity $m$ takes roughly $O(m^2)$ time both for the discrete and the continuous version~\cite{alt1995computing, agarwal2014computing}.
This is optimal under the strong exponential time hypothesis~\cite{B14,abboud2018tighter}, even for 1-dimensional curves~\cite{bringmann2016approximability,buchin2019seth} and it holds even if we allow a (small) constant approximation.

Related to our problem is a variant of the discrete Fréchet distance, which optimizes over all possible translations 
\cite{mosig2005approximately,jiang2008protein,avraham2015faster,bringmann2021discrete}. Here, the state of the art is by Bringmann et al.~\cite{bringmann2021discrete} who give a $\Tilde{O}(m^{4.66\ldots})$ algorithm and complement this result with a lower bound of $m^{4-\eps}$ conditioned on the strong exponential time hypothesis.

Closely related to our work is the problem of clustering time series data. Driemel et al. \cite{DKS16} defined the $(k,l)$-clustering problem for time series as follows: Given a set $P$ of $n$ time series of complexity at most $m$ and parameters $k,l \in \NN$ find $k$ center time series of complexity at most $l$, such that (a) the maximum distance of an element in $P$ to its closest center time series or (b) the sum of these distances is minimized. Variant (a) is referred to  as $(k,l)$-center and (b) as $(k,l)$-median. Under the continuous Fréchet distance, they developed near linear time $(1+\eps)$-approximation algorithms for both clustering variants, assuming $\eps,k$, and $l$ are constants. They complement these algorithmic results with hardness results, showing that both $(k,l)$-median and $(k,l)$-center are NP-hard under continuous Fréchet distance. Approximating $(k,l)$-median for polygonal curves in arbitrary dimensions was recently studied by Buchin et al. \cite{BDR23}. Cheng and Huang give the first $(1+\eps)$-approximation algorithm for $(k,l)$-median under continuous Fréchet distance in $d > 1$~\cite{CH23}. Both clustering problems are also NP-hard under the discrete Fréchet distance and even for the case $k = 1$ \cite{BDGHKLS19} \cite{buchin2019hardness}. Nath and Taylor \cite{nath2020k} then developed a $(1+\eps)$-approximation algorithm for $(k,l)$-median under discrete Fréchet distance, which runs in polynomial time. Driemel et al. recently developed a near-linear time approximation scheme for  $(k,l)$-median clustering under discrete Fréchet, assuming that $k$ and $l$ are constant \cite{driemel2025nearlineartimeapproximationscheme}.

Buchin and Rohde \cite{buchin2022coresets} designed the first coreset construction for $(k,l)$-median under both variants of the Fréchet distance, where the size of the coreset has logarithmic dependence on the number of input curves.
Recently, Cohen-Addad et al. introduced a new technique for coresets  yielding an input-size–independent coreset for $(k,l)$-median~\cite{CDRSS25}.

\subsection{Results \& High Level Approach}
The main conceptual contribution of this paper is the introduction of the Fréchet Decomposition Problem (Problem \ref{prob: decomposition}, Section~\ref{sec:problem}) for which we obtain two main algorithmic results 
that hold for $n$ time series of complexity at most $m$,   approximation parameter $\varepsilon>0$ and
 complexity $l$ of the $k$ base curves. In the following, we assume that $l,k$ and $\varepsilon$ are constants. 
\begin{itemize}
\item
a $(1+\varepsilon)$-approximation
algorithm with $\tilde O(n^2 + nm)$ running time for the Fréchet Decomposition Problem for $k=1$ (Corollary \ref{cor:main1}), and
\item 
a $(1+\varepsilon)$-approximation algorithm with $\tilde O(nm + |C|^k n)$ running time
for the Fréchet Decomposition Problem when the set of base curves is taken from a given finite set of candidate curves $C$ for constant $k\ge 1$ 
(Corollary \ref{cor:main2}).
\end{itemize}

 We consider the setting that $k=1$ in Section~\ref{sec:constantfactor}. Here we give a constant-factor approximation algorithm with running time $\Tilde{O}(n^2m)$ for Problem \ref{prob: decomposition}  
(Theorem~\ref{thm: constant approx k equals 1}) and extend this algorithm to a $(1+\eps)$-approximation algorithm with running time in $ \tilde{O}(nm^{4.5} +n^2m)$ 
(Theorem~\ref{thm:epsapprox}). To achieve this result, we exploit that the problem definition allows us to control the infinity norms of the base curves and input time series, which is expressed by the key Lemmas \ref{lem: scale base time series} and \ref{lem: scale single input time series}. For the case $k = 1$, this allows us to show that either weight $1$ or $-1$ is a constant approximation with respect to the optimal weight if the base curve and input time series are normalized appropriately.
We then compute the minimum-error $l$-simplifications of the (appropriately normalized) input time series and choose the simplification that minimizes a simplified variant of our objective function where weights are restricted to $\{-1,1\}$ to get a constant approximation for the Fréchet decomposition problem with $k=1$. This is then further refined to a $(1+\varepsilon)$-approximation by discretizing the area around the vertices of the constant approximation with an appropriately chosen grid.

In Section~\ref{sec:exact}, we give 
an exact algorithm for the Fréchet Projection Problem (Problem \ref{prob: distance}, Section~\ref{sec:problem}) that computes
the cost of a set of time series with respect to a given set of base curves. To do so, we start with a dynamic programming approach to solve a variant of Problem \ref{prob: distance}, where the weights are fixed (Lemma \ref{lem: unweighted projection error}). To extend this to the weighted setting, we construct an arrangement of hyperplanes such that for each weight vector in the same cell of the arrangement, it holds that the dynamic program takes the same sequence of decisions.
Given a cell, we can then construct a linear program, based on the corresponding sequence of decisions in the dynamic program, to find the best weights in the given cell. By doing so for each cell in the arrangement, we obtain an
algorithm that has running time $O( m^{2k+2.5})$ (Theorem \ref{lem: weights and traversal matrices}).
This algorithm then leads to an exact algorithm for 
the Fréchet Decomposition Problem in the setting where the set of base curves must come from a given finite set $C$. The running time of the algorithm is in $O(|C|^k n m^{2k+2.5})$  (Corollary~\ref{cor:projection}). 

We combine the results of both sections with a dimension reduction result of Driemel et al. (Theorem 2 in \cite{driemel2025nearlineartimeapproximationscheme})  that allows us to improve the dependency on $m$ in the running times.

In Appendix \ref{sec: greedy} we discuss a natural greedy approach for the Fréchet Decomposition Problem and show that the approximation factor of such an approach is unbounded.

\section{Preliminaries}

We represent a time series of length $m$ as an $m$-dimensional vector $x \in \RR^m$ over the reals and refer to $m$ as the \emph{complexity} of the time series.

In the following, we will introduce the discrete Fréchet distance. To do so we start with the standard definition that uses the notion of traversals. 
 
\begin{definition}[traversal]
Given two time series $x = (x_1,\dots,x_{m}) \in \RR^{m}$ and $y = (y_1,\dots,y_{l}) \in \RR^{l}$ a \emph{traversal} $T$  is a sequence of index pairs $(i,j) \in [m]\times[l]$ satisfying:
\begin{enumerate}
\item $T$ starts with pair $(1,1)$
\item $T$ ends with pair $(m,l)$ 
\item if $(i,j)$  appears in $T$, it can only be followed by either $(i,j+1), (i+1,j)$ or $(i+1,j+1)$.
\end{enumerate}
If a pair $(i,j)$ appears in $T$, we say that $x_i$ and $y_j$ are \emph{matched} in traversal $T$.
\end{definition}
Next, we state an alternative notation that uses matrices to model the traversals, which was introduced in \cite{driemel2025nearlineartimeapproximationscheme}. Here, a traversal is encoded by special matrices that map an $l$-dimensional vector $x$ to a $t$-dimensional space by copying some of the entries of $v$ while keeping their relative order. These traversal matrices are defined as follows.

\begin{definition}[traversal matrices]
$M \in \{0,1\}^{t \times l}$ is a traversal matrix if it satisfies: 

\begin{enumerate}
    \item The rows of $M$ are standard unit vectors,
    \item $M(1,1) = 1$,
    \item $M(t,l) = 1$, and
    \item For rows $i \in [t-1]$ and columns $ j \in [l-1]$ it holds that if $M(i,j) = 1$ then either $M(i+1,j+1) = 1$ or $M(i+1,j) = 1$.  
    If $i\in [t-1]$ and $j=\ell$ then $M(i+1,j)=1$. 
\end{enumerate}
\end{definition}

For $t,l \in \NN$ we denote with $\mathcal{M}^t_l$ the set of all traversal matrices of dimensions $(t,l)$. 
A pair of traversal matrices 
$(M_m, M_l)$
has \emph{matching dimension}
for two time series $x\in \mathbb R^m$ and $y\in \mathbb R^l$, if
there is $t \in \mathbb N$ such that $t\ge\max(m,l)$,
$M_m \in \mathcal M_{m}^t$, and $M_l\in \mathcal M_{l}^t$.
When the dimensions $t$ and $l$ are clear from the context, then we will not mention $\mathcal{M}^t_l$ explicitly. Note that
one can represent any traversal of $x=(x_1,\dots, x_m)$ and $y=(y_1,\dots, y_l)$ by a pair of traversal matrices
$(M_m, M_l)$ of matching dimension. 
Additionally the $t$-dimensional vectors 
$M_m \cdot (1,2,\dots, m)^T$ and $M_l \cdot (1,2,\dots, \ell)^T$
denote the index of the entry
of $x$ and $y$ that is mapped to the corresponding 
entry in the $t$-dimensional vectors $M_m x$ and $M_l y$, respectively. This allows us to identify $(M_m, M_l)$ with a sequence $T'$ of index pairs whose $r$-th entry
is the pair $(i,j)$ where $i$
is the $r$-th entry of 
$M_m \cdot (1,2,\dots, m)^T$ and $j$ is the $r$-th entry of $M_l \cdot (1,2,\dots, l)$. One obtains a traversal by removing multiple neighboring occurrences of pairs $(i,j)$ from the sequence $T'$. For a traversal matrix $M$ of dimensions $(t,m)$, we will call the vector $M \cdot (1,2,\dots,m)^T$ the \emph{traversal identifier} of $M$.

\begin{definition} [Discrete Fr\'echet distance]
Given time series $x \in \RR^m,y \in \RR^l$ the 
 \emph{discrete Fréchet distance} $d_F(x,y)$ between $x$ and $y$ is defined as
    \[d_F(x,y) = \min \left \Vert M_m x - M_l  y \right \Vert_\infty,\]
where the minimum is over all pairs of traversal matrices $M_m, M_l$ with matching 
dimension. As a standard extension we define for any set $Y \subseteq \bigcup_{l \in \NN} \RR^l$
\[d_F(x,Y) = \min_{y\in Y} d_F(x,y).\] 
\end{definition}

We introduce the notion of the Fréchet span of a set of time series $Y =\{y_1,\dots, y_k\}$. The span consists of all summations of time series that are formed from $y_1,\dots,y_k$ by applying traversal matrices of matching dimension to each $y_i$.

\begin{definition} [Fréchet closure]
    Let $y \in \RR^l$, the Fréchet $t$-closure of $y$ is defined as
\[
\closure_t(y) =  \left\{ M y ~|~  M \in \mathcal{M}^t_l \right\}.
\] 
It is the set of all Fréchet expansions to vectors of length $t$.
The closure is the union of the $t$-closures over all $t \in \NN$.
\end{definition}

\begin{observation}\label{obs:frechetclosureinftydist}
The discrete Fr\'echet distance between $x$ and $y$ can be written as the minimum $\ell_{\infty}$ distance between the two closures of $x$ and $y$:
\[
d_F(x,y) = \min_{t \in \NN} \left\{~ \| \hat{x}-\hat{y} \|_{\infty}    ~|~  \hat{x} \in \closure_t(x), \hat{y} \in \closure_t(y)\right\}
\]
\end{observation}

\begin{definition} [Fréchet span]
For $Y=\{y_1,\dots, y_k\} \subset \RR^l$, we define the Fréchet span of $Y$ as the set of linear combinations of the elements of the closures of $y_i \in Y$.
\[
\spann(Y) = \bigcup_{w \in \RR^k} \bigcup_{t \in \NN } \left\{ \sum_{i=1}^k w_i x_i ~|~ x_i 
\in \closure_t(y_i) \right\}
\] 
\end{definition}

Later in the paper, we will use the concept of minimum-error $l$-simplification, which can be computed in $O(l  m  \log(m)  \log(m/l))$ time as in \cite{bereg2008simplifying} and is defined as follows.

\begin{definition}[minimum-error $l$-simplification]

A 
time series $x' \in \RR^l$ is a minimum-error $l$-simplification of 
a time series $x \in \RR^m$, 
 
if for any time series $y \in \RR^l$, 
$\df(x,x') \leq \df(x,y)$.
\end{definition}

We will further utilize the following theorem of Driemel et al.~\cite{driemel2025nearlineartimeapproximationscheme} to reduce the complexity of the input time series to improve the running times of our results.

\begin{theorem}[\cite{driemel2025nearlineartimeapproximationscheme}]\label{thm: dim_reduction}
Let $\varepsilon \in (0,1)$ and $ l\in \NN$ be constants. There exists a constant $d_0 = d_0(\varepsilon, l)$ such that for arbitrary $m\in \NN$ and input time series $x \in \RR^m$ one can compute in time $O(m \log^2 m)$ 
a time series $z \in \RR^{d_0}$ such that for all $y \in \RR^l$,
\begin{equation*}
    (1-\varepsilon) d_{F}(x,y) \leq d_{F}(z,y) \leq (1 + \varepsilon)  d_{F}(x,y).
\end{equation*}
\end{theorem}

\textbf{Hyperplane arrangements.} A \emph{hyperplane} is a set $\{x \in \RR^d : \langle a, x\rangle = b\}$ where $a \in \RR^d \setminus \{0\}$ and $b \in \RR$. A \emph{half-space} is a set $\{x \in \RR^d : \langle a, x\rangle \leq b\}$ where $a \in \RR^d \setminus \{0\}$ and $b \in \RR$, and an \emph{open half-space}
is a set $\{x \in \RR^d : \langle a, x\rangle < b\}$ where $a \in \RR^d \setminus \{0\}$ and $b \in \RR$. An \emph{arrangement} $\AAA(H)$ of a finite set $H$ of hyperplanes in $\RR^d$ is a partition
of $\RR^d$ into cells: a $d$-dimensional cell in $\AAA(H)$ is a maximal connected
region of $\RR^d$
not intersected by any hyperplane in $H$. Unless otherwise stated, we refer to $d$-dimensional cells as cells. A cell $F$ of $\AAA(H)$ can be described
by a sign vector in $\{-1,+1\}^{|H|}$ which is defined with respect to a fixed orientation of each hyperplane: each $h\in H$ partitions $\RR^d$ into three regions: $h$, and the two
open half-spaces determined by it; one of them is arbitrarily assigned to $+1$ and the other one is assigned to $-1$. For a cell $F$ of an arrangement $\AAA(H)$, $\mathrm{cl}(F)$ denotes its closure.

\section{Problem statement}
\label{sec:problem}
Our goal is to decompose our set of input time series into base curves (plus error). For a set of base curves $B=\{b_1,\dots, b_k\}\subset \RR^l$, the error associated with an input curve $x$ is equal to the distance between $x$ and the Fr\'echet span of $B$, $\spann(B)$. Our goal is to minimize the sum of such errors over all input curves. This is formalized as follows.

\begin{problem} [Fr\'echet Decomposition] \label{prob: decomposition}
Let $k,l\ge 1$ be integers.
Given a set $P \subset \RR^m$ of $n$ time series,
find a set $B \subset \RR^l$ of $k$ base curves that minimize
\[
\cost(P,B) = \sum_{x \in P}  d_F(x,\spann(B)).
\]

\end{problem}

A closely related, but likely simpler problem asks for the ``projection'' of a given input time series $x$ to a fixed set of base curves. 
\begin{problem}[Fr\'echet Projection Distance]\label{prob: distance}
   Let $x\in \RR^m$ be a time series and $B=\{b_1,\dots, b_k\}$ be a set of time series from $\RR^l$, where $k$ and $l$ are fixed integers. The projection distance problem is to compute $d_F(x,\spann(B))$.

\end{problem}

\section{$(1+ \varepsilon)$-approximation for finding a single base curve $(k=1)$}\label{sec:epsapprox}\label{sec:constantfactor}

In this section, we present an algorithm achieving a $(1+\varepsilon)$-approximation for the Fr\'echet Decomposition problem with $k=1$. Our approach begins by computing a constant approximation, which is subsequently refined to obtain a $(1+\varepsilon)$-approximation.

We begin our exposition with the algorithm that computes a constant approximation. \Cref{alg: constantApproxK1} first computes minimum-error $l$-simplifications of the (appropriately normalized) input time series and chooses the simplification that minimizes a simplified variant of our objective function where weights are restricted to $\{-1,1\}$. The first goal of this section is to prove correctness and bound the running time of \Cref{alg: constantApproxK1}.

\begin{algorithm}[H]
\caption{\label{alg: constantApproxK1}}
\begin{algorithmic}[1]
\Procedure{ \texttt{ConstantApprox} }
{$ P \subset \RR^m$, $l \in \NN$}
\For{\textbf{each} $x \in P$} 
\State$y_x \gets$ minimum-error $l$-simplification of $\frac{x}{\Vert x \Vert_\infty}$
\EndFor
\State $min\gets \infty$
\For{\textbf{each} $x \in P$} 
\State $c\gets  \sum_{z\in P}\Vert z \Vert_\infty \cdot\min_{w\in \{-1,1\}} d_F(\frac{z}{\Vert z \Vert_\infty},w\cdot y_x)$ 
\If{$min >c$}
\State $min \gets c$
\State $y\gets y_x$
\EndIf
\EndFor
\State \Return $y$
\EndProcedure
\end{algorithmic}
\end{algorithm}

 Missing proofs of this section can be found in  \Cref{app:constantfactor}. 
We begin with two lemmas that show the effect on the cost after either scaling the base curves or the input time series. Note that while these are stated for general $k$, we will only use them for $k=1$. 

\begin{restatable}{lemma}{lemmarescaledbasis}
\label{lem: scale base time series}
Given time series $x \in \RR^m$ and $B= \{b_1,\dots,b_k\} \subset \RR^l$ with $l \leq m$. Let $\beta  \in (\RR \setminus\{0\})^k$ and $B' = \{\beta_1 b_1,\dots,\beta_k b_k\}$ then it holds that
\[
d_F(x,\spann(B)) =  d_F(x,\spann(B')).
\]
\end{restatable}

\begin{proof}
\begin{align*}
d_F(x,\spann(B))
&=   \min_{\substack{w_1,\dots, w_k \in \RR \\ M_0,M_1,\dots,M_k}} \Vert  M_0 \cdot x - \sum^k_{j=1} w_j \cdot M_j \cdot b_j \Vert_\infty\\
&= 
\min_{\substack{w_1,\dots, w_k \in \RR \\ M_0,M_1,\dots,M_k}} \Vert  M_0 \cdot x - \sum^k_{j=1} w_j \cdot M_j \cdot \frac{\beta_j}{\beta_j} \cdot b_j \Vert_\infty \\
&= 
\min_{\substack{w_1,\dots, w_k \in \RR \\ M_0,M_1,\dots,M_k}} \Vert  M_0 \cdot x - \sum^k_{j=1} \frac{w_j}{\beta_j} \cdot M_j \cdot \beta_j \cdot b_j \Vert_\infty  \\
&=  \min_{\substack{w'_1,\dots, w'_k \in \RR \\ M_0,M_1,\dots,M_k }} \Vert  M_0 \cdot x - \sum^k_{j=1} w'_j \cdot M_j \cdot \beta_j \cdot b_j \Vert_\infty\\
&= d_F(x,\spann(B'))
\qedhere
\end{align*}
\end{proof}

\begin{restatable}{lemma}{lemmascaledcurve}
    \label{lem: scale single input time series}
Given time series $x \in \RR^m$ and 
$B = \{b_1,\dots,b_k\} \subset  \RR^l$ with $l \leq m$. For any $\alpha \in \RR \setminus\{0\}$ it holds that
\[d_F(x,\spann(B)) = \frac{1}{|\alpha|}  d_F(\alpha \cdot  x,\spann(B)).\]
\end{restatable}
\begin{proof}
\begin{align*}
d_F(x,\spann(B))
&= \min_{\substack{w_1,\dots, w_k \in \RR \\ M_0, M_1,\dots,M_k }} \Vert  M_0 \cdot  x - \sum^k_{j=1} w_j \cdot M_j \cdot b_j \Vert_\infty \\
&= 
\min_{\substack{w_1,\dots, w_k \in \RR \\ M_0, M_1,\dots,M_k }} \Vert  \frac{\alpha}{\alpha} \cdot M_0 \cdot x - \sum^k_{j=1} w_j \cdot M_j \cdot \frac{\alpha}{\alpha} \cdot b_j \Vert_\infty\\
&=
\frac{1}{|\alpha|}  \min_{\substack{w_1,\dots, w_k \in \RR \\ M_0,M_1,\dots,M_k }} \Vert  \alpha \cdot M_0 \cdot  x - \sum^k_{j=1} w_j \cdot \alpha \cdot M_j \cdot b_j \Vert_\infty \\
&= \frac{1}{|\alpha|}  \min_{\substack{w'_1,\dots, w'_k \in \RR \\ M_0, M_1, \dots,M_k }} \Vert  \alpha \cdot M_0 \cdot  x - \sum^k_{j=1} w'_j \cdot M_j \cdot b_j \Vert_\infty \\
&= \frac{1}{|\alpha|} d_F(\alpha \cdot x,\spann(B))
\qedhere
\end{align*}
\end{proof}

We will utilize Lemma \ref{lem: scale base time series} and Lemma \ref{lem: scale single input time series} to assume without loss of generality that all input time series and base curves have their infinity norm realized by a non-negative value, since we can always multiply them by $-1$ without changing the cost. Additionally, by Lemma \ref{lem: scale base time series}, we will assume in the following that the base curves have $\ell_{\infty}$ norm $1$.

Next, we make an observation on lower and upper bounds of the discrete Fréchet distance by identifying the value pairs that have the smallest and largest differences. 
In the following, for $v=(v_1,\ldots,v_r)\in \RR^r$ we will write $\max(v) = \max_{1\le i \le r} v_i$ and $\min(v) = \min_{1\le i \le r} v_i $.  

\begin{obs}\label{obs: sign_dependant_bounds}
For time series $x \in \RR^m$, $b \in \RR^l$ and $w \in \RR$ it holds that
  \begin{itemize}
      \item $d_F(x,w b) \leq \max\{|\max(x) - \min(wb)|, |\min(x) - \max(wb)|\}$, and 
      \item $d_F(x,wb)\geq \max\{|\max(x) - \max(wb)|,|\min(x) - \min(wb)|\}$. 
  \end{itemize}

\end{obs}
Note that since a time series $v$ gets flipped by multiplying it with a negative weight $w\in \RR_{<0}$, it holds $w\max(v) = \min(wv)$ and $w\min(v) = \max(wv)$.

To show that Algorithm~\ref{alg: constantApproxK1} is correct,  we first consider the problem of finding the weight that minimizes $d_F(x,w\cdot b)$ for fixed  $x \in \RR^m $ and $b \in \RR^l$.
Note that if $\Vert b \Vert_\infty = 0$, any weight minimizes $d_F(x,w\cdot b)$ regardless of $x$. On the other hand, if $\Vert x \Vert_\infty =0$  then weight $0$ minimizes $d_F(x,w\cdot b)$, since $x$ has to be all zeroes. We will therefore only consider time series $x$ and $b$ with $\Vert x\Vert_\infty, \Vert b \Vert_\infty \neq 0$.

Now we show that any optimal weight is contained in an interval with its boundaries depending on $\Vert x \Vert_\infty$ and $\Vert b \Vert_\infty$. 
The proof can be found in \Cref{app:constantfactor}.

\begin{restatable}{lemma}{lemmaboundweights}\label{lem:bound_weights}
For time series $x \in \RR^m$ with $\Vert x \Vert_\infty = \max(x) $ and $b \in \RR^l$ with $\Vert b \Vert_\infty = \max(b) $ it holds that 
\[\argmin_{w \in \RR} d_F(x,w\cdot b) \subseteq \left[\frac{-2\Vert x \Vert_\infty}{\Vert b \Vert_\infty},\frac{2\Vert x \Vert_\infty}{\Vert b \Vert_\infty}\right].\] 
\end{restatable}

Next, we show that the weight $\frac{\Vert x \Vert_\infty}{\Vert b \Vert_\infty}$ gives a constant approximation for the best non-negative weight. 
The proof is also diverted to \Cref{app:constantfactor}. 
\begin{restatable}{lemma}{lemmaapproxpositiveweights}\label{approx_positive_weights}
    For time series $x \in \RR^m$ with $\Vert x \Vert_\infty = \max(x) $ and $b \in \RR^l$ with $\Vert b \Vert_\infty = \max(b) $ it holds that
 \[\df\left(x, \frac{\Vert x \Vert_\infty}{\Vert b \Vert_\infty} \cdot b\right) \leq 2\min_{w \in \RR_{\geq 0}}\df(x,w \cdot b).\]
\end{restatable}

Next, we investigate non-positive weights. Here, we first show that an additive error can occur when a positive weight is used instead of a corresponding negative weight.
\begin{restatable}{lemma}{lemmaadditiveerrorswitchingweightsign}\label{additive_error_switching_weight_sign}
    Given time series $x\in \RR^m$, $b\in \RR^l$, and  a non-positive weight $w \in \RR_{\leq 0}$, it holds that 
\[\df(x,|w|\cdot b) \leq \df(x,w \cdot b) + 2 \cdot |w| \cdot \Vert b \Vert_\infty.\] 
\end{restatable}

We can now show that using either $-1$ or $1$ as weight is a constant approximation for any weight if $x$ and $b$ are scaled in a certain way through Lemmas \ref{lem: scale base time series} and \ref{lem: scale single input time series}.

\begin{restatable}{lemma}{lemmaapproxweights}\label{lem: approx_weights}
Given time series $x\in \RR^m$ and $b\in \RR^l$ 
such that $\Vert x \Vert_\infty  = \Vert b \Vert_\infty = \max(x) = \max(b)$, it holds that
\[\min_{w \in \{-1,1\}} \df(x, w\cdot b) \leq 10\min_{w \in \RR}\df(x,w \cdot b).\]
\end{restatable}

By scaling the input time series to satisfy the condition of Lemma \ref{lem: approx_weights} we have to deal with weighted distances. The next lemma shows that  we can find a constant factor approximation among the simplifications of the scaled input time series with respect to a modified variant of the objective function.

\begin{restatable}{lemma}{lemmaapproxconstantweights}\label{lemma:approx_constant_weights}
     Given a set of time series $\{x_1,\dots, x_n\}\subset \RR^m$, parameters $a_1, \dots, a_n \in \RR \setminus\{0\}$, and $l \in \NN$, let $b' \in \argmin_{b\in \RR^l}\sum_{i =1}^n |a_i| \min_{w\in \{-1,1\}} \df \left( \frac{x_i}{a_i},w \cdot b \right)$ and $y_j$ be a minimum-error $l$-simplification of $\frac{x_j}{a_j}$, where $j \in \argmin_{i \in [n]} \min_{w\in \{-1,1\}}\df\left(\frac{x_i}{a_i},w \cdot b'\right)$. Then, \[\sum_{i=1}^n |a_i| \min_{w \in \{-1,1\}} \df\left(\frac{x_i}{a_i}, w \cdot y_j\right) \leq 3 \sum_{i=1}^n |a_i|  \min_{w \in \{-1,1\}} \df\left( \frac{x_i}{a_i}, w \cdot b'\right).\]
\end{restatable}

We are now ready to analyze  \Cref{alg: constantApproxK1}.
We use the fact that Algorithm \ref{alg: constantApproxK1} normalizes all the input time series to the same $\ell_\infty$ norm, which allows us to prove the following theorem by combining  Lemma \ref{lem: approx_weights} and Lemma \ref{lemma:approx_constant_weights}.

\begin{restatable}{theorem}{thmsconstantapprxkequalone}\label{thm: constant approx k equals 1}
    Given  $n$ time series $P \subset \RR^m$ and $l \in \NN$ then Algorithm \ref{alg: constantApproxK1} computes in time $\Tilde{O}(n^2ml)$ a $30$-approximation for the Fr\'echet decomposition problem with $k = 1$.
\end{restatable}

\begin{proof}
Let $b^* \in \argmin_{b \in \RR^l} \sum_{x \in P} \min_{w \in \RR} \df(x,w \cdot b)$ be an optimal solution for the Fr\'echet decomposition problem with $k = 1$. By Lemma \ref{lem: scale base time series} we can assume that $\Vert b^* \Vert_\infty = \max(b^*) = 1$.  Let $y$ be the time series returned by the algorithm. We show that $y$ is a $30$-approximation as follows. 
By \Cref{lem: scale single input time series}, we obtain

\[
\sum_{x \in P}  \min_{w\in\RR} \df(x,w \cdot y) = \sum_{x \in P} \Vert x \Vert_\infty \min_{w\in\RR} \df(\frac{x}{\Vert x \Vert_\infty},w \cdot y)\leq \sum_{x \in P} \Vert x \Vert_\infty \min_{w \in \{-1,1\}} \df(\frac{x}{\Vert x \Vert_\infty},w \cdot y).
\]
Next, by  \Cref{lemma:approx_constant_weights},
\begin{align*}
\sum_{x \in P}  \min_{w\in\RR} \df(x,w \cdot y)
&\leq 3 \min_{b \in \RR^l} \sum_{x \in P}\Vert x \Vert_\infty \min_{ w \in \{-1,1\}} \df(\frac{x}{\Vert x \Vert_\infty},w \cdot b)\\
&\leq 3 \sum_{x \in P}\Vert x \Vert_\infty \min_{ w \in \{-1,1\}} \df(\frac{x}{\Vert x \Vert_\infty},w \cdot b^*),
\end{align*}
and, by \Cref{lem: approx_weights} and \Cref{lem: scale single input time series},

\[
\sum_{x \in P}  \min_{w\in\RR} \df(x,w \cdot  y)
\leq 30 \sum_{x \in P} \Vert x \Vert_\infty \min_{w \in \RR} \df(\frac{x}{\Vert x \Vert_\infty},w \cdot  b^*)= 30 \sum_{x \in P}  \min_{w \in \RR} \df(x,w \cdot  b^*).
\]

The minimum-error $l$-simplifications of the time series can be computed in time 
$\tilde{O}(nl m)$   \cite{bereg2008simplifying}. Finding the best time series among the simplifications takes $O(n^2 m l)$ time.
\end{proof}

Our next step is to show that we can further improve the quality of the solution to an approximation ratio of $1+\varepsilon$. 
We show now that for an optimal base curve $b^{\ast}$, any time series inducing a small cost is either near  $b^{\ast}$ or near $-b^{\ast}$.

\begin{restatable}{lemma}{lemoptimalisnear}\label{lem:optimalisnear}
Given a set of $n$ time series $P \subset \RR^m$, 
let $b^* \in \RR^l$ be an optimal base curve and $b' \in \RR^l$ be a time series with cost  $\Delta'$.
Then,
\[\min_{\substack{w\in \{-1,1\} } }\df( b',   w\cdot b^*) \leq \frac{40 \Delta'}{\sum_{x\in P} \|x\|_{\infty}}. \]
\end{restatable}

Now, given a time series $b'\in \RR^{l}$ which is either near to $b^{\ast}$ or $-b^{\ast}$, we can discretize the area around its vertices by an appropriately chosen grid and enumerate all possible near time series; one of them is guaranteed to be sufficiently close to $b^{\ast}$ (or $-b^{\ast})$, therefore providing an approximate solution whose approximation factor is controlled by the side-length parameter of the grid.
To compute the cost of a time series we make use of the following claim, which is implied by \Cref{lem: weights and traversal matrices} in Section \ref{ssec:exactprojection}.
\begin{claim}\label{clm: computeCostkEq1} Given a set of time series $P \subset \RR^m$ and a time series $b \in \RR^l$, one can compute $cost(P,\{b\})$ in time $ O(m^{4.5}l^2)$.

\end{claim}
This procedure is described in Algorithm \ref{alg: epsApproxK1} and its performance is stated in \Cref{thm:epsapprox}. 

\begin{algorithm}[H]
\caption{\label{alg: epsApproxK1}}
\begin{algorithmic}[1]
\Procedure{ \texttt{EpsilonApprox} }
{$ P \subset \RR^m$, $l \in \NN$, $\varepsilon \in \RR_{>0}$}
\State $c \gets \texttt{ConstantApprox}(P,l)$ \Comment{\Cref{alg: constantApproxK1}}

\State $\Delta' \gets \cost(P,\{c\})$
\State $X \gets \sum_{x\in P} \Vert x \Vert_\infty$
\State $\delta \gets \frac{\varepsilon}{60} \cdot \frac{\Delta'}{X}$
\State $S_i \gets [c_i - \frac{40 \Delta'}{X}, c_i + \frac{40 \Delta'}{X} ]$ for $i \in [l]$ 
\State $\mathcal{C} \gets (\bigcup_{i \in [l]} S_i \cap G_\delta)$, where $G_\delta = \{i \cdot \delta ~|~ i \in \ZZ\}$
\State Let $b' \in \argmin_{b \in \mathcal{C}^l} \cost(P,\{b\})$ be arbitrary
\State \Return $b'$
\EndProcedure
\end{algorithmic}
\end{algorithm}

\begin{restatable}{theorem}{thmepsapprox}\label{thm:epsapprox}
    Given $n$ time series $P\subset \RR^m$ and parameter $l \in \NN$, $\varepsilon \in \RR_{>0}$, let $b^*$ be an optimal solution for the Fréchet decomposition problem with $k = 1$. Then
    Algorithm \ref{alg: epsApproxK1} computes in time $O\left( \frac{l}{\varepsilon}\right)^l \cdot O(nm^{4.5} l^2)+\tilde{O}(n^2ml)$ a time series $b' \in \RR^l$ such that 
 
    \[
    \cost(P,\{b'\})\leq (1+\varepsilon)\cdot \cost(P,\{b^{\ast}\}).
    \]
\end{restatable}

\begin{proof}
By Theorem \ref{thm: constant approx k equals 1} it holds that $\Delta' \in [\cost(P,\{b^*\}), 30 \cost(P,\{b^*\})]$ and ConstantApprox$(P,l)$ computes $c$  in time $\Tilde{O}(n^2ml)$. 

By \Cref{lem:optimalisnear},  there exists $w\in \{-1,1\}$, such that
    \[\df(  c,   w\cdot b^*) \leq \frac{40 \Delta'}{\sum_{x\in P} \|x\|_{\infty}}. \]
    Hence, all vertices of $b^{\ast}$ or $-b^{\ast}$ lie in intervals $S_i:=\left[c_i -\frac{40 \Delta'}{\sum_{x\in P} \|x\|_{\infty}} , c_i +\frac{40 \Delta'}{\sum_{x\in P} \|x\|_{\infty}} \right]$, 
  
    where $c_i$, $i\in [l]$ are the vertices of $c$. 
    Let $\mathcal{C}$ be defined as in the algorithm, which is the area of interest discretized by a grid with side length $\delta$. 

    The algorithm returns the time series $b'$ in $\mathcal{C}^l$ that minimizes the cost. 
    \Cref{lem:optimalisnear} and the construction of $\mathcal{C}$ imply that there exists a time series $b\in\mathcal{C}^l$ such that $\|b - {b}^{\ast}\|_{\infty}\leq \delta$. For $x \in P$, let $w_x^{\ast} \in \argmin_{w\in \RR} \df(x,w\cdot b^{\ast})$.
    The cost of $b'$ is then upper bounded as follows:
    \begin{align}
        \cost(P,\{b'\}) &\leq \cost(P,\{b\})\\
        &=\sum_{x\in P}\min_{w\in\RR}\df(x,w\cdot b)\nonumber\\
        &\leq \sum_{x\in P}\df(x,w_x^{\ast}\cdot  b) \nonumber\\
        &\leq \sum_{x\in P}\left(\df\left(x,w_x^{\ast}\cdot {b^{\ast}}\right)+\df\left(w_x^{\ast}\cdot b^{\ast},w_x^{\ast}\cdot b\right) \right)\label{eq:epsapprox1}\\
         &\leq \cost(P,\{b^{\ast}\}) + \sum_{x\in P}w_x^{\ast} \delta \nonumber\\
          &\leq \cost(P,\{b^{\ast}\}) + 2\delta\sum_{x\in P}\|x\|_{\infty}   \label{eq:epsapprox2} \\
          &\leq \cost(P,\{b^*\}) + \frac{\varepsilon}{30}\Delta'\\
          &\leq (1+\varepsilon)\cost(P,\{b^{\ast}\}),\nonumber
    \end{align}
    where in \eqref{eq:epsapprox1} we use the triangle inequality and in \eqref{eq:epsapprox2} we use \Cref{lem:bound_weights}.

    For each time series in $\mathcal{C}$, we compute the total cost in time $O(nm^{4.5} l^2)$ using \Cref{clm: computeCostkEq1}. The size of $\mathcal{C}$ is in $O\left( \frac{l}{\varepsilon}\right)^l$; hence, the overall running time is $O\left( \frac{l}{\varepsilon}\right)^l \cdot O(nm^{4.5} l^2) +\tilde{O}(n^2ml)$. 
\end{proof}

By first applying Theorem \ref{thm: dim_reduction} on the input time series we obtain the following Corollary.
\begin{corollary}
\label{cor:main1}
    Given $n$ time series $P\subset \RR^m$ and constants $l \in \NN$ and $\varepsilon \in (0,1)$. Let $b^*$ be an optimal solution for the Fréchet decomposition problem with $k = 1$. One can compute in time $\tilde O(n^2+nm)$ a time series $b' \in \RR^l$ such that 
    \[
    \cost(P,\{b'\})\leq (1+\varepsilon)\cdot \cost(P,\{b^{\ast}\}).
    \]
\end{corollary}

\section{Exact algorithms}
\label{sec:exact}

\subsection{Projection Distance Problem}
\label{ssec:exactprojection}
In this section, we give an algorithm for 
the projection distance problem. Missing proofs of this section are located in Appendix \ref{app:exact}. We start by solving the problem for the case that the weight vector is fixed. Note that in this case we can assume without loss of generality that the weight vector is the vector of all $1$s since we can scale each base curve by its optimal weight and use the resulting set of base curves instead.
We solve the problem using dynamic programming.   

Let $x \in \RR^m$ and $b_1,\dots, b_k \in \RR^l$.
For a vector $v=(v_1,\dots, v_r)$ we will use $v^{(i)} = (v_1,\dots, v_i)$ to be the $i$-dimensional vector that agrees with $v$ on the first $i$ coordinates.  Additionally for $b_j$ with $1\leq j \leq k$ we use the notation $(b_j)_i$ to address its $i$-th coordinate. 
For the dynamic program, we use $D \subset \RR^{m\times l^k}$ to denote a multidimensional array of which each entry $D[i_0,i_1,\dots, i_k]$ will store the target value
\[
\min_{M_0,\dots,M_k} \Vert  M_0 x^{(i_0)} - \sum^k_{j=1}  M_j b_j^{(i_j)}\Vert_\infty
\]
Intuitively, the target value is the Fréchet distance of the prefix of $x$ to the (unit weight) span of the prefixes of the base vectors (cf.\ Observation~\ref{obs:frechetclosureinftydist}). The subproblems solved by the dynamic program are parametrized by the lengths of these prefixes. To define the dynamic program computing these entries, we use the following recursion: 
\[
D[i_0,\dots, i_k] = \max \big\{\min_{(h_0,\dots,h_k) \in H_{i_0,\dots,i_k}} D[i_0 - h_0, \dots,i_k - h_k ], |x_{i_0}-\sum^k_{j=1}(b_j)_{i_j}| \big\} , 
\]
where $H_{i_0,\dots,i_k} = \big\{(h_0,\dots,h_k)\in \{0,1\}^{k+1}~:~ i_j -h_j \geq 1 \text{ for } 0 \leq j \leq k \big\} \setminus \big\{(0,\dots,0)\big\}$ is the set of possible movements among the involved time series 
and we set $D[1, \ldots , 1] = |x_1 - \sum^k_{j=1} (b_j)_1 |$. We may abbreviate notation and write  $D[\vec{i}]$ and $H_{\vec{i}}$ for $\vec{i} = ( i_0, \dots,i_k)$. 
The dynamic program is given in Algorithm \ref{alg: 
projection}. The proof builds upon the reasoning behind the standard DP formulation of the discrete Fréchet distance computation in \cite{eiter1994computing}.

\begin{algorithm}
\caption{\label{alg: projection}}
\begin{algorithmic}[1]
\Procedure{ \texttt{Projection} }
{$ x \in \RR^m, \{b_1, \dots , b_k \} \subset \RR^l$}
\State Initialize $D \subset \RR^{m \times l^k}$  \label{line: init table}
\For{$\vec{i} \in [m] \times [l]^k$ in lexicographical order }\label{line: loop positions}
\If{$\vec{i} = (1,\dots, 1)$}
\State $D[\vec{i}] \gets |x_1 - \sum^k_{j=1} (b_j)_1|$
\Else
\State $D[\vec{i}] \gets \max \big\{\min_{h \in H_{\vec{i}}} D[\vec{i}-h], |x_{i_0} - \sum^k_{j=1} (b_j)_{i_j}| \big\}$ \label{line: dp_decision}
\EndIf
\EndFor
\State \Return $D$
\EndProcedure
\end{algorithmic}
\end{algorithm}

Given the computed array $D$ of Algorithm \ref{alg: projection} one can reconstruct optimal traversal matrices with a standard dynamic programming technique by starting from the top-right corner of $D$ and tracing back to the origin, which decisions have led to the solution. Here we only compute the optimal traversal matrices implicitly by constructing their traversal identifiers to save space and time, since the explicit matrix construction is not necessary for our purposes.
The running time of this procedure is dominated by the computation of the table $D$ yielding the following theorem. The full proof of correctness is in Appendix \ref{app:exact}.

\begin{restatable}{theorem}{lemmatraversalscomputation}\label{lem: compute traversal matrices}
Given time series $x \in \RR^m$, $b_1,\ldots, b_k\in \RR^l$, one can  compute in $O((2l)^{k}    m )$ time,  traversal identifiers of traversal matrices $M_0,\dots,M_k$ that minimize $\Vert  M_0 x - \sum^k_{j=1}  M_jb_j\Vert_\infty$.
\end{restatable}

We now solve the projection distance problem when the weights are not fixed.

\begin{theorem}{}\label{lem: weights and traversal matrices}
Given a time series $x \in \RR^m$ and $k$ time series $B = \{b_1,\dots, b_k\} \subset \RR^l$, one can compute $d_F(x,\spann(B))$ 
in $O(2^{k}  m^{(2k+2.5)} l^{2k^3})$ time. 

\end{theorem}

\begin{proof}
    Our algorithm initially constructs an arrangement of hyperplanes, such that each cell of the arrangement defines a subset of $\RR^k$ corresponding to weight values that all lead to the same sequence of decisions taken by \Cref{alg: projection}. Having such an arrangement then allows us to treat our problem in two steps: we compute the best traversal matrices within each cell, and for each set of such traversal matrices we compute the best possible weights.  
    Let $M_{0}^*,\ldots, M_{k}^*$ be optimal traversal matrices and $w^*_1,\dots,w^*_k \in \RR$ be optimal weights:
\[ d_F(x,\spann(B)) = \Vert M^*_0 x - \sum^k_{j=1} w^*_j M^*_j b_j\Vert_\infty.\]

We start by defining hyperplanes that encode relations of distances between any vertex of the input time series and any candidate (matched) vertex of its projection. 
Given $\vec{i} = (i_0,\dots,i_{k})\in [m]\times [l]^k, \vec{j}=(j_0,\dots,j_{k})\in [m]\times [l]^k$, for each $\sigma\in\{-1,+1\}$ we define   
\[
s_{\sigma}(\vec{i},\vec{j}):=\left\{(w_1,\ldots,w_k)\in \RR^{k} :
 x_{i_0}- \sum_{y=1}^{k}w_{y }(b_y)_{i_
y}   =  \sigma\cdot( x_{j_0}- \sum_{y=1}^{k}w_{y }(b_y)_{j_
y}  )
\right\}
\]
and
$s(\vec{i},\vec{j}):=\left\{s_{\sigma}(\vec{i},\vec{j})\mid \sigma\in \{-1,+1\} \right\}$. 
Essentially, the arrangement of hyperplanes 
$\AAA(s(\vec{i},\vec{j}))$
partitions $\RR^k$ into cells, such that any weight vector within the same cell induces the same sign on $x_{i_0}- \sum_{y=1}^{k}w_{y }(b_y)_{i_
y}  -  \sigma\cdot( x_{j_0}- \sum_{y=1}^{k}w_{y }(b_y)_{j_
y})$, for any $\sigma \in \{-1,+1\}$, therefore inducing the same sign on  $|x_{i_0}- \sum_{y=1}^{k}w_{y }(b_y)_{i_
y}|  -  | x_{j_0}- \sum_{y=1}^{k}w_{y }(b_y)_{j_
y}|$. 

Next, we define the set of hyperplanes 
 \[\mathcal{S}:= \bigcup_{\vec{i},\vec{j}\in [m]\times [l]^k} s(\vec{i},\vec{j}).\] 
For each cell $F$ of the arrangement $\AAA(\mathcal{S})$ there exists a common ordering of the values $|x_{i_0}- \sum_{y=1}^{k}w_{y }(b_y)_{i_
y}|$ over all $i \in [m] \times [l]^k$. This implies that each cell $F$ of the arrangement $\AAA(\mathcal{S})$ is a subset of $\RR^k$, for which, for any $(w_0,\dots,w_k) \in \mathrm{cl}(F)$, the evaluations made by \Cref{alg: projection} in Line \ref{line: dp_decision}, throughout all iterations, are fixed. Hence, by the correctness of \Cref{alg: projection}, 
there are traversal matrices $M_0, \dots,M_{k}$ such that 
for any $(w_0,\dots,w_k) \in \mathrm{cl}(F)$,
$M_0, \dots,M_{k}$ minimize $ \Vert M_0 x - \sum_{i=1}^k w_i M_{i} b_i\Vert_{\infty}$.
Additionally, by the definition of $\mathcal{S}$  there is a cell $F^*$ of $\AAA(\mathcal{S})$ such that~$M_{0}^*,\dots, M_{k}^*$ minimizes $\Vert M_0 x  - \sum_{i=1}^k w_i M_{i}  b_i \Vert_{\infty}$ for any $(w_0,\dots,w_k) \in \mathrm{cl}(F^*)$.
We then compute the minimizing weights of $\mathrm{cl}(F)$ as follows. Let $c_0,\ldots,c_k$  be the traversal identifiers of \Cref{lem: compute traversal matrices}, let $t$ be their common dimension, and let $M_0,\dots,M_k$ be the corresponding optimal traversal matrices.
Finally let
$I_F:= \{((c_0)_r,\ldots,(c_k)_r) \mid \text{ for } r \in [t] \}$. We have
\[
\forall w\in \mathrm{cl}(F):~  \Vert M_0 x - \sum_{i=1}^k w_i M_{i}    b_i \Vert_\infty = \max_{(i_0,\ldots,i_{k})\in I_F} | x_{i_0}-\sum_{j=1}^{k}  w_j  (b_j) _{i_j}  |.
\]
Moreover, by the definition of $\mathcal{S}$, for each cell $F$ of $\AAA(\mathcal{S})$, there is a uniquely defined function $\sigma_F:~[m]\times[l]^k \mapsto \{-1,1\}$ such that
\[
\forall w\in \mathrm{cl}(F):~\Vert M_0 x - \sum_{i=1}^k M_{i} w_j b_i \Vert_{\infty}:=\max_{(i_0,\ldots,i_{k})\in I_F} \sigma_F(i_0,\ldots,i_{k})\cdot (x_{i_0}-\sum_{j=1}^{k} w_j  (b_j)_{i_j}  ) .
\]
Formally, we need to solve the following optimization problem:  \[\min_{w \in \mathrm{cl}(F)} \max_{(i_0,\ldots,i_{k})\in I_F} \sigma_F(i_0,\ldots,i_{k})\cdot (x_{i_0}-\sum_{j=1}^{k}  w_j (b_j)_{i_j}  ).\] 
This can be written as a linear program as follows: 
\begin{align*}
    \text{objective:} &~\min_{r\in \RR} r \\
\text{constraints:} & ~\forall (i_0,\ldots,i_{k})\in I_F:~     \sigma_F(i_0,\ldots,i_{k}) \cdot (x_{i_0}-\sum_{j=1}^{k} w_j (b_j)_{i_j} )\leq r\\
&~(w_1,\ldots,w_k)\in \mathrm{cl}(F).
\end{align*}
Doing so for each cell $F$ of $\AAA(\mathcal{S})$, we return the best weights $w$ and $M_0, \dots,M_{k}$ encountered this way, implying the correctness. Note that the linear program is bounded, since by the definition of $\sigma_F$, the left-hand side terms of the constraint inequalities are non-negative; hence $r$ is non-negative.

Next, we analyze the running time of this procedure. The arrangement $\AAA(\mathcal{S})$ contains $O(  m^2l^{2k})$ hyperplanes. Therefore, the number of cells is in $O(  m^{2k}  l^{{2k}^2})$ and $\mathcal{S}$ can be constructed using the Algorithm of Edelsbrunner,  O’Rourke, and Seidel~\cite{EOS86} in time $O(  m^{2k}  l^{{2k}^2})$. For each cell we 
apply Corollary \ref{lem: compute traversal matrices}, which requires $O(2^k m  l^k)$ time.
The running time needed to find optimal weights of a cell is in $O((km)^{2.5})$ by \cite{vaidya1989speeding} since the LP contains $O(km)$ variables. 
In total, this is $O( m^{2k}  l^{2k^2}(2^km  l^k + (k m)^{2.5})) \subset O(2^{k}  m^{(2k+2.5)} l^{2k^3})$.
\end{proof}

\subsection{Fr\'echet Decomposition Problem}
\label{ssection:exactdecomposition}
The results of \cref{ssec:exactprojection} imply a solution for the Fr\'echet decomposition problem when the base curves are chosen from a predefined finite set of candidates $C$. By simply enumerating all $k$-combinations from $C$, and computing their costs using \Cref{lem: weights and traversal matrices}, we obtain the following result.  
\begin{corollary}\label{cor:projection} Let $k$ and $l\in \NN$ be constants.
Given a set of $n$ time series $P \subset \RR^m$, $k \in \NN$, a set of time series $C \subset \RR^l$ with $|C|\geq k$, we can compute a set of $k$ time series $B \subseteq C$ that minimizes
$cost(P,B)$ in time ${O}\left(|C|^k   n m^{2k+2.5}\right)$.  
\end{corollary}

By first applying Theorem \ref{thm: dim_reduction} on the input, we obtain the following corollary.

\begin{corollary}\label{cor:main2}
Let $k,l\in \NN$ and $\varepsilon \in (0,1)$ be constants.
Given a set of $n$ time series $P \subset \RR^m$, a set of time series $C \subset \RR^l$ with $|C|\geq k$ and $B \subseteq C$ with $|B| = k$. We can compute a set of $k$ time series that approximates
$cost(P,B)$ within a factor of $(1+\varepsilon)$ in time $\tilde{O}\left(nm+|C|^k   n \right)$, where the $\tilde{O}$-notation hides terms that are polylogarithmic.  
\end{corollary}

 \bibliography{SourceFiles/biblio} 
 \bibliographystyle{plainurl}

\newpage
\appendix

\section{Experimental validation of the problem definition}\label{sec:experiments}
To demonstrate the effect of the time series decomposition problem we conducted an experiment to compare it to the principal component analysis.

\subsection{Data set}
 We used the bike sharing data set of the UCI Machine Learning Repository \cite{bike_sharing_275}. The data set corresponds to a historical log  that originated from Capital Bikeshare system, Washington D.C., USA, and contains the number of casually rented bikes for each day in the years 2011 and 2012. Since bike sharing highly correlates with environmental and seasonal settings, i.e. day of the week, weather and season, we utilize this data to visualize the effect of the time series decomposition problem.

\begin{figure}[h]
  \centering
    \includegraphics[scale = 0.5]{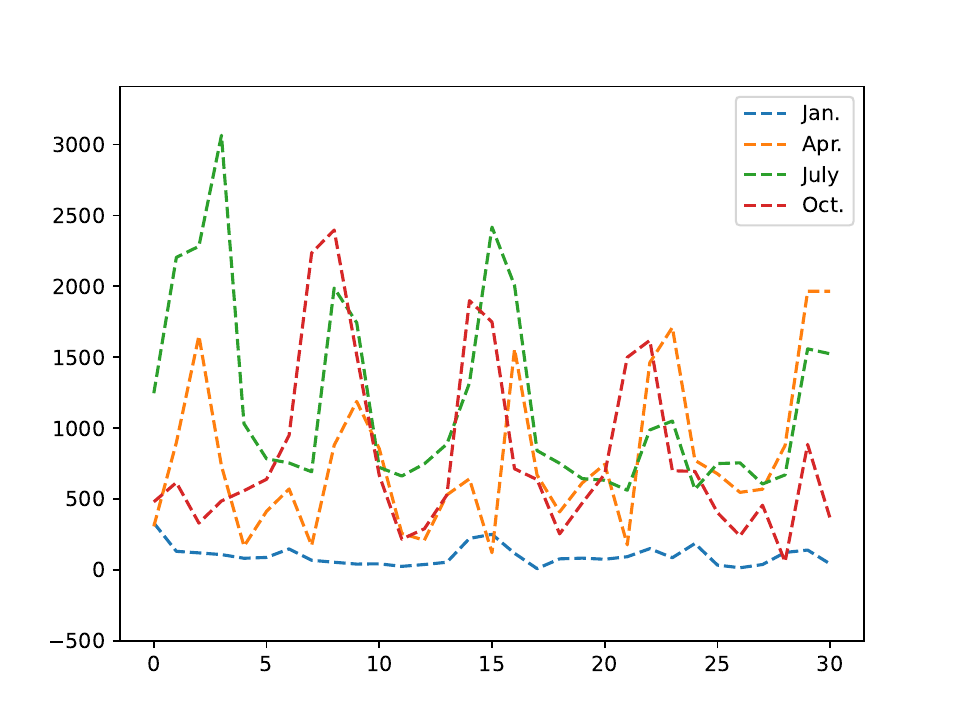}
     \caption{Casual rented bikes for sample months of the year 2011.}
     \label{fig:1}
\end{figure}

\subsection{Experimental setup}
 For the experimental setup we converted the aggregated data of casually rented bikes to time series for each month involved in the given timespan, refer to Figure~\ref{fig:1}. We heuristically solve the time series decomposition problem for $k=1$ on this input.
The algorithm is closely related to the work of Uebbing et al. \cite{uebbing2024alternative} (\hyperlink{https://github.com/jahoec/PCA-on-Curves}{https://github.com/jahoec/PCA-on-Curves}).
As in Section \ref{ssection:exactdecomposition} we restrict the choice of base curves to a predefined set of candidate time series.
This set of time series consists of approximate minimum error $\ell$-simplifications for all input time series and all values $\ell \in[2,31]$. We used a grid search (over the interval $[-2,2]$ with discrete steps of size $0.1$) to approximate the weights in the projection distance problem, thereby avoiding the costly arrangement computation.
Finally we consider a one sided variant of the discrete Fréchet distance, meaning that we only apply a traversal matrix to the base curve. Formally we consider the objective function $\sum_{x\in P} \min_{\substack{M, \\ w \in \RR}} \Vert x - M\cdot w \cdot b\Vert_\infty$. The reason for this restriction is to make the expanded base curve interpretable in the original timespan of the input time series.
For comparison we also applied principal component analysis on the time series. To make the principal component analysis applicable we extended all time series to the common complexity $31$ by duplicating their last vertex.
The base curve and principal component are both computed on the whole set of $24$ input time series.

The principal component analysis is computed as follows. Let $P  \in \RR^{24\times 31}$ be the given input time series represented as a matrix where each row corresponds to an individual month. Let $z \in \RR^{31}$ be the feature wise mean vector, i.e. for $i \in [31]$ let $z_i = \frac{1}{24}\sum_{j=1}^{24} P_{ji}$. Then let $P' = P - z$ and apply a singular value decomposition on $P'$, i.e. $P' = U\cdot \Sigma \cdot V^T$. The first principal component is given by scaling the first right singular vector with the first singular value, i.e.  $pc=\Sigma_{1,1}\cdot (V^T)_{1}$. We then get the projection of the $i$-th input time series by multiplying the principal component with the corresponding score given in the left singular vector and adding back the mean $z$, i.e. $x'_i = U_{1,i}\cdot pc + z$.

\subsection{Results}
Figure \ref{fig:2} shows the computed base curve and the reconstructed input time series and Figure \ref{fig:3} shows the first principal component and the reconstructed input time series by only using the first principal component. For visualization purposes only a sample of the input time series is plotted. 

The total error of both procedures is taken through the ratio $\sum_{x\in P} \Vert x - x'\Vert_2^2 / \sum_{x\in P} \Vert x\Vert_2^2$, where $x'$ denotes the reconstruction of $x$. The time series decomposition approach achieves a total error of $\approx14 \%$, while the principal component analysis approach has a total error of $\approx23 \% $. As such, the new approach clearly outperforms the PCA approach.

\begin{figure}[h]

    \begin{subfigure}[t]{0.48\textwidth}
        \centering
       \includegraphics[scale = 0.4]{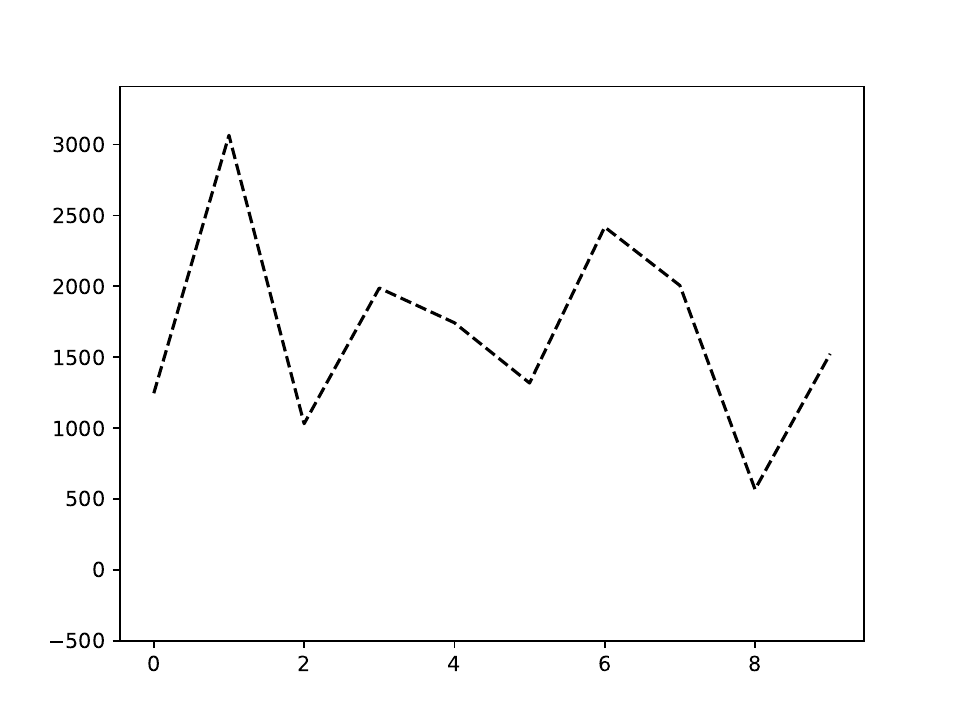}
        \label{fig:2.a}
         
    \end{subfigure} 
    \begin{subfigure}[t]{0.48\textwidth}
        \centering
       \includegraphics[scale = 0.4]{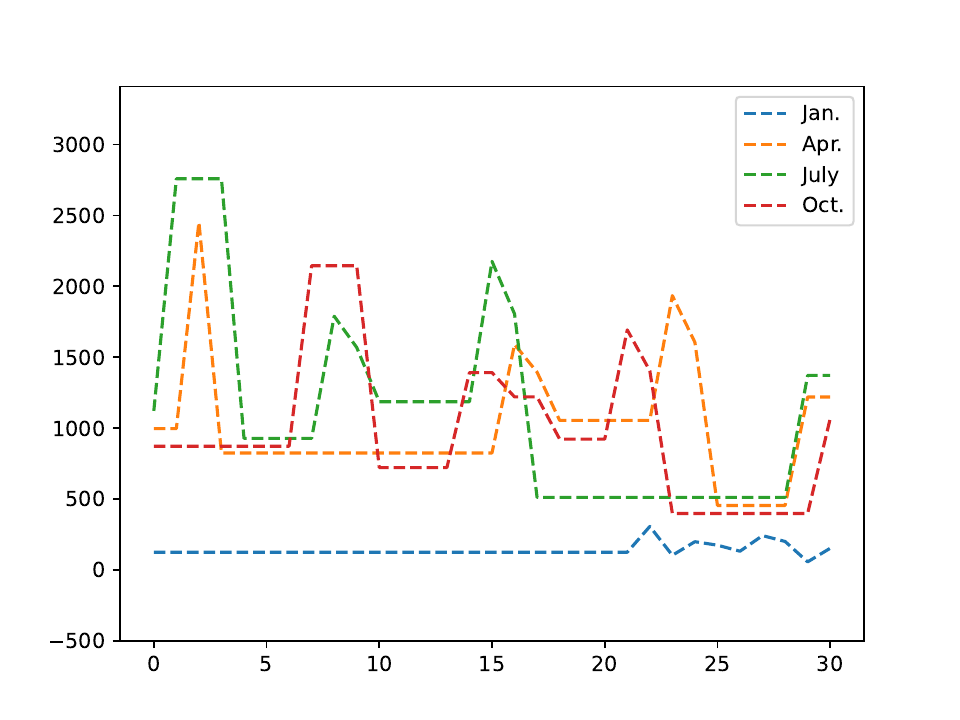}
         \label{fig:2.b}
    \end{subfigure} 
   
  \caption{The left plot shows the computed base curve. The right plot shows the reconstructed input time series by multiplying the corresponding weight and traversal matrix to the base curve.
  }
  \label{fig:2}
  \end{figure}
  
\begin{figure}[h]
\begin{subfigure}[t]{0.48\textwidth}
        \centering
       \includegraphics[scale = 0.4]{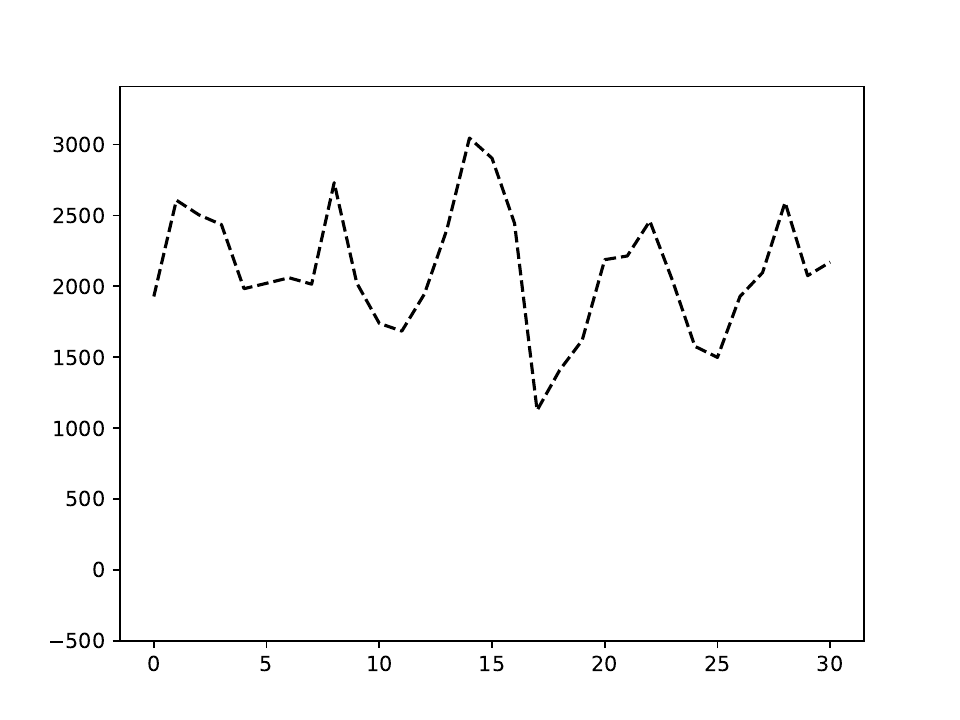}
         \label{fig:3.a}
    \end{subfigure} 
    \begin{subfigure}[t]{0.48\textwidth}
        \centering
       \includegraphics[scale = 0.4]{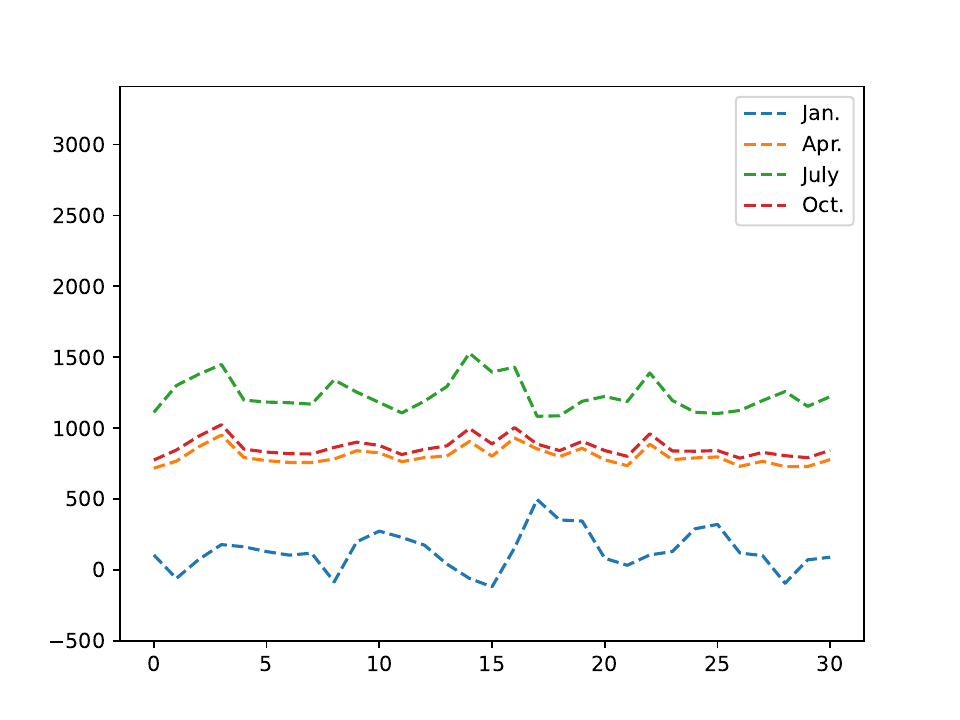}
         \label{fig:3.b}
    \end{subfigure} 
      \caption{
      The left plot shows the first principal component. The right plot shows the reconstruction of the time series by restricting to the first principal component.
     }   
     \label{fig:3}
\end{figure}

\newpage
\section{A (failed) greedy approach} \label{sec: greedy}
A natural approach to the Fr\'echet Decomposition problem is a greedy strategy that solves the problem iteratively by fixing the base curves computed in previous iterations and only optimizes one base curve at a time. Here, we first compute
\[b_1 \in \argmin_{b\in \RR^l} \sum_{x\in P} \min_{w \in \RR} d_F(x,w\cdot b) \] and continue with $k$-iterations, where in iteration $i$ we fix the previously computed base curves $b_1,\dots,b_{i-1}$ and compute \[b_i \in \argmin_{b\in  \RR^l} \sum_{x\in P} \min_{\substack{ w_1,\dots,w_i \\M_0,M_1,\dots, M_i}} \| M_0 x - \sum_{j=1}^{i-1} M_j w_jb_j   - M_i w_i b\|_{\infty} .\]

Consider the time series $x = (1,1,1)^T$ and $y = (1,0,1)^T$ as input, and let $l = 2$ be the complexity of the base curves.
For $k=2$ optimal base curves are $b_1^*= (1,0)^T$ and $b_2^*=(0,1)^T$, since there are traversal matrices $M_1,M_2$ and $M'_2$ such that
\[M_1b_1^* + M_2b_2^* =  (1,0,0)^T + (0,0,1)^T  = (1,0,1)^T\] and 
\[M_1b_1^* + M'_2b_2^* = (1,0,0)^T  + (0,1,1)^T = (1,1,1)^T.\]
Therefore, the optimal cost is $0$.

Next, we investigate the greedy strategy that first computes an optimal base curve for $k=1$ and then optimizes the second base curve by fixing the first one.
For $k = 1$ we argue that the set $A=\{(a,a)^T ~\mid~ a \in \RR\}$ contains all optimal base curves. For an arbitrary $b_1 \in A$ we have  $\min_{w\in{\RR}} d_F(x,wb_1) = 0$ and $\min_{w\in \RR} d_F(y,wb_1) = \frac{1}{2}$. Now assume a time series $(a,c)^T$ with $a \neq c$. Then $\min_{w\in{\RR}} d_F(x,w(a,c)^T) > 0$ and $\min_{w\in{\RR}} d_F(y,w(a,c)^T) \geq \frac{1}{2}$, therefore incurring larger cost than $b_1$. Now observe that for $k=2$ there is no time series of complexity $2$ that can construct input time series $y$ in combination with $b_1$, implying that the cost of the greedy strategy is larger than $0$.

\section{Missing Proofs of Section \ref{sec:constantfactor}}
\label{app:constantfactor}

\lemmaboundweights*

\begin{proof}
It holds that \[\min_{w \in \RR_{}} d_F(x,w\cdot b) \leq d_F(x,0\cdot b) = \Vert x \Vert_\infty.\]
 
Consider some arbitrary $w^* \in \argmin_{w \in \RR} \df(x,w\cdot b)$.
First assume $w^* < \frac{-2\Vert x \Vert_\infty}{\Vert b \Vert_\infty}$ then for some $\varepsilon > 0$ we get
\begin{align*}
\df\left(x,w^*\cdot b\right) &=
\df\left(x,\left( \frac{-2\Vert x \Vert_\infty}{\Vert b \Vert_\infty}-\varepsilon\right)\cdot b\right)\\
&\geq \left|\min(x) - \min \left( -1\left(\frac{2\Vert x \Vert_\infty}{\Vert b \Vert_\infty} + \varepsilon\right)\cdot  b \right )\right|\\
&\geq \left|\min(x)+\left(\frac{2\Vert x \Vert_\infty}{\Vert b \Vert_\infty} + \varepsilon\right) \max(b)\right|\\
&= \min(x) + 2\Vert x \Vert_\infty + \varepsilon \Vert b \Vert_\infty\\
&\geq\Vert x \Vert_\infty + \varepsilon \Vert b \Vert_\infty\\
&>\Vert x \Vert_\infty,
\end{align*} where the first inequality holds by Observation \ref{obs: sign_dependant_bounds}.
Next consider some $w^* > \frac{2\Vert x \Vert_\infty}{\Vert b \Vert_\infty}$. For some $\varepsilon > 0$ we get
\begin{align*}
\df\left(x,w^*\cdot b\right) &= \df\left(x,\left(\frac{2\Vert x \Vert_\infty}{\Vert b \Vert_\infty}+\varepsilon\right)\cdot b\right)\\
&\geq \left|\max(x)-\frac{2\Vert x \Vert_\infty}{\Vert b \Vert_\infty}\max(b)-\varepsilon\max(b)\right|\\
&= \left | -\Vert x \Vert_\infty - \varepsilon \Vert b \Vert_\infty \right |\\
&=\Vert x \Vert_\infty + \varepsilon \Vert b \Vert_\infty \\
&> \Vert x \Vert_\infty.
\end{align*}
Both assumptions on $w^*$ lead to a contradiction and therefore $ \frac{-2\Vert x \Vert_\infty}{\Vert b \Vert_\infty} \leq w^* \leq \frac{2\Vert x \Vert_\infty}{\Vert b \Vert_\infty} $.

\end{proof}

\lemmaapproxpositiveweights*

\begin{proof}
Let $w^* \in \argmin_{w \in \RR_{\geq 0}}\df(x,w \cdot b)$, $\Delta := \df(x,w^*\cdot b)$ and $\alpha \in \RR$ such that $w^* = \frac{\Vert x \Vert_\infty}{\Vert b \Vert_\infty} + \alpha$.
By Observation \ref{obs: sign_dependant_bounds} we get the following inequality: 
\begin{align*}
\Delta &=  \df(x, w^* \cdot b )\\
&=\df \left(x, \left(\frac{\Vert x \Vert_\infty}{\Vert b \Vert_\infty}+ \alpha \right)\cdot b \right) \\
&\geq \left| \max(x)  - \left(\frac{\Vert x \Vert_\infty}{\Vert b \Vert_\infty} + \alpha \right)\max(b)\right|\\
&= \left|\Vert x \Vert_\infty - \left(\frac{\Vert x \Vert_\infty}{\Vert b \Vert_\infty} + \alpha \right)\Vert b \Vert_\infty  \right|\\
&= |\alpha| \Vert b \Vert_\infty.
\end{align*}
We further get by the triangle inequality
\begin{align*}
d_F\left(x,\frac{\Vert x \Vert_\infty}{\Vert b\Vert_\infty} \cdot  b\right) 
&= d_F(x,(w^* -  \alpha) \cdot b) \\
&\leq d_F(x,w^* \cdot b) + |\alpha|\Vert b\Vert_\infty\\
&\leq 2   \df(x,w^* \cdot b).
\end{align*}

\end{proof}

\lemmaadditiveerrorswitchingweightsign*

\begin{proof}
By Observation \ref{obs: sign_dependant_bounds} it  holds that \[\df(x,|w|\cdot b) \leq \max \{|\max(x) -  |w|  \min(b)|, |\min(x) - |w|  \max(b)|\}\]
and
\[\df(x,w\cdot b) \geq \max \{|\max(x) +  |w|  \min(b)|, |\min(x) + |w|  \max(b)|\}.\]
We have 
\begin{align}
&|\max(x) - |w|  \min(b)|= \nonumber\\
&= |\max(x) - |w|  \min(b)| - |\max(x) + |w|  \min(b)| + |\max(x) + |w|  \min(b)| \nonumber\\
&\leq |\max(x) - |w| \min(b)| - |\max(x) + |w| \min(b)| + \df(x,w \cdot b)\nonumber\\
&\leq |\max(x) - |w| \min(b) - ( \max(x) + |w|   \min(b))| + \df(x,w \cdot b)\nonumber\\
&= 2\cdot |w| \cdot |\min(b)| + \df(x,w \cdot b).\label{eq:lem12a}
\end{align}
and
\begin{align}
&|\min(x) - |w| \max(b)|=\nonumber\\
&= |\min(x) - |w| \max(b)| - |\min(x) + |w|  \max(b)| + |\min(x) + |w|  \max(b)|\nonumber\\
&\leq |\min(x) - |w|  \max(b)| - |\min(x) + |w|  \max(b)| + \df(x,w \cdot b)\nonumber\\
&\leq |\min(x) - |w|  \max(b)  - (\min(x)  + |w|  \max(b) )| + \df(x,w \cdot b)\nonumber\\
&\leq 2 \cdot |w| \cdot |\max(b)|+ \df(x,w \cdot b). \label{eq:lem12b}
\end{align}
Therefore, by \eqref{eq:lem12a} and \eqref{eq:lem12b}, we have \[\df(x,|w|\cdot b) \leq 2\cdot |w| \cdot  \Vert b \Vert_\infty + \df(x,w \cdot b).\]
\end{proof}

\lemmaapproxweights*

\begin{proof}
Let $\beta \in \RR_{\geq 0}$ such that $\beta = \Vert x \Vert_\infty = \max(x)  = \Vert b \Vert_\infty = \max(b)$.
If there exists some $ w^* \in \argmin_{w \in \RR}\df(x,w \cdot b)$
with $w^* \geq 0$, then the proof follows by Lemma \ref{approx_positive_weights}. Therefore assume that this is not the case and consider some arbitrary $  w^* \in \argmin_{w \in \RR}\df(x,w \cdot b)$. By Lemma \ref{lem:bound_weights} we have $0 \geq w^* \geq -2$.
First, assume $-1 > w^* \geq -2$. Then there is some $\alpha \in (0,1]$ with $w^* = -1-\alpha$.
This gives the following inequality:
\begin{align*}
\Delta &:= \df(x,w^* \cdot b)\\
&\geq |\min(x) - w^* \max(b)|\\
&=|\min(x) - (-1 - \alpha) \max(b)|\\
&=|\min(x) + \beta + \alpha \beta|\\
&\geq \alpha \beta,
\end{align*}
where the first inequality holds by Observation~\ref{obs: sign_dependant_bounds} and the last inequality holds since $\min(x) \geq - \beta$.
We further get by triangle inequality \[\df(x,-1 \cdot b) \leq \df(x,w^* \cdot b) + |\alpha|\beta  \leq 2 \df(x,w^* \cdot  b).\] 

For the rest of the proof consider the case $0 > w^*\geq -1$. Then there is $\alpha\in [0,1)$ with $w^* = -1+\alpha$.  Assume  that $-\beta \leq \min(b) \leq -\beta/2$, which leads to the following inequality:
\begin{align*}
\Delta &:= \df(x,w^* \cdot  b)\\
&\geq |\max(x) - w^* \min(b)|\\
&=|\beta - w^* \min(b)|\\
&= |\beta - (-1+\alpha)\min(b)|\\
&= |\beta+ \min(b) - \alpha \min(b)|\\
&= \beta - |\min(b)| + \alpha |\min(b)|\\
&\geq \alpha|\min(b)|\\
&\geq \frac{\alpha \beta}{2 }.
\end{align*}
We further get by triangle inequality
\[\df(x,-1 \cdot b) \leq \df(x,w^* \cdot b) + |\alpha|\beta \leq  3  \df(x,w^* \cdot b).\] 
Finally, consider the case $\min(b) \geq - \beta/2$. Here, by the case assumptions, we have that $w^*\cdot \min(b) \in [-\beta, \beta/2]$ and therefore $\df(x,w^* \cdot  b )\geq |\max(x) - w^*\min(b)|\geq \beta/2$.
By Lemma \ref{approx_positive_weights} and Lemma \ref{additive_error_switching_weight_sign}  we now get
\begin{align*}
\df(x,b) &\leq 2 \min_{w \in \RR_{\geq 0}} \df(x,w \cdot  b)\\
&\leq 2\df(x,|w^*| \cdot   b)\\
&\leq 2 (\df(x,w^* \cdot  b) + 2|w^*|\beta)\\
&\leq 2(\df(x,w^* \cdot  b) + 2\beta)\\
&= 2(\df(x,w^*  \cdot b) + 4/2\beta)\\
&\leq 2(\df(x,w^*\cdot   b) + 4 \cdot \df(x,w^* \cdot  b))\\
&\leq 10  \df(x,w^* \cdot  b).
\end{align*}\end{proof}

\lemmaapproxconstantweights*

\begin{proof}
Let $w_i \in \argmin_{w \in \{-1,1\}} \df \left( \frac{x_i}{a_i},w  \cdot  b' \right)$ be arbitrary for $ i \in [n]$. Consider the following inequalities:
\begin{align*}
&\sum_{i = 1}^n |a_i| \min_{w \in \{-1,1\}} \df\left(\frac{x_i}{a_i}, w  \cdot  y_j\right) \leq  \sum_{i = 1}^n |a_i|  \df\left(\frac{x_i}{a_i}, \frac{w_i}{w_j}  \cdot y_j\right)\\
\intertext{(By triangle inequality)}
&\leq \sum_{i=1}^n |a_i| \left(  \df\left(\frac{x_i}{a_i}, \frac{w_i}{w_j} \cdot \frac{x_j}{a_j}\right) + \df\left(\frac{w_i}{w_j} \cdot  \frac{x_j}{a_j}, \frac{w_i}{w_j} \cdot y_j\right) \right)\\
&\leq \sum_{i = 1}^n |a_i| \left (  \df\left(\frac{x_i}{a_i}, w_i \cdot  b'\right) + \df\left(w_i \cdot b', \frac{w_i}{w_j} \cdot \frac{x_j}{a_j}\right)  + \df\left(\frac{w_i}{w_j}\cdot  \frac{x_j}{a_j},\frac{w_i}{w_j} \cdot y_j\right) \right)\\
&=\sum_{i = 1}^n |a_i| \left(  \df\left( \frac{x_i}{a_i}, w_i \cdot b'\right) +   \left|\frac{w_i}{w_j}\right| \left(\df\left(w_j \cdot  b',  \frac{x_j}{a_j}\right)  + \df\left( \frac{x_j}{a_j}, y_j\right) \right) \right)\\
&=\sum_{i = 1}^n |a_i| \left(  \df\left( \frac{x_i}{a_i}, w_i \cdot b'\right) +   \df\left(w_j \cdot  b',  \frac{x_j}{a_j}\right)  + \df\left( \frac{x_j}{a_j}, y_j\right)\right)\\
\intertext{(Since $y_j$ is a minimum-error $l$-simplification of $\frac{x_j}{a_j}$)}
&\leq \sum_{i = 1}^n |a_i| \left(\df\left( \frac{x_i}{a_i}, w_i \cdot  b'\right) +  2 \df\left(w_j\cdot  b',  \frac{x_j}{a_j}\right) \right)\\
\intertext{(By our choice of $j$)}
&\leq 3 \sum_{i = 1}^n |a_i| \cdot  \df\left( \frac{x_i}{a_i}, w_i \cdot  b'\right),
\end{align*} which concludes the proof.
\end{proof}

\lemoptimalisnear*

\begin{proof} For $x \in P$ let $x' = \frac{x}{\Vert x\Vert_\infty}$,  $P':=\left\{ x' \mid x \in P \right\}$ and let $\Delta^{\ast}$ be the cost of the optimal base curve $b^{\ast}$. 
For $b \in \{b',b^*\}$ we can, by Lemma \ref{lem: scale base time series}, assume that $\Vert b \Vert_\infty = 1$  and by Lemma  \ref{lem: approx_weights}, for any $x'\in P'$, we have:   \[ {\min_{w\in\{-1,1\}} \df(x',w\cdot b)} \leq 10\min_{w\in \RR}\df(x',w \cdot b).
\]
We will prove our statement by contradiction. Therefore, we assume that 
\begin{align}
    \min_{\substack{w\in \{-1,1\} } }\df(b',   w\cdot b^*) > \frac{40 \Delta'}{\sum_{x\in P} \|x\|_{\infty}}.\label{eq:contradictionassumption}
\end{align}
Now let $W(P')$ be the distribution with support $P'$, where each $x' \in P'$ is sampled with probability $\frac{\|x\|_{\infty}}{\sum_{y\in P}\|y\|_{\infty}}$. Notice that for any $b\in\RR^l$,  \begin{align}
    \bE_{x'\sim W(P')}\left[\min_{w\in \RR}\df(x',w \cdot b)\right] &= \frac{\sum_{x\in P} \|x\|_{\infty}\min_{w\in \RR}\df\left(\frac{x}{\|x\|_{\infty}},w\cdot b \right) }{\sum_{x\in P} \|x\|_{\infty}} \\
    &= 
    \frac{\sum_{x\in P}  \min_{w\in\RR}\df\left(x,w\cdot b \right) }{\sum_{x\in P} \|x\|_{\infty}} ,\label{eq:contradictionexpectation}
\end{align}
where the last equality holds by \Cref{lem: scale single input time series}. 
Let $C := \{x'\in P' \mid {\min_{w \in \{-1,1\}} \df(x',w \cdot b')} < \frac{20 \Delta'}{\sum_{y\in P}\|y\|_{\infty}} \}$.
By Markov's inequality and \eqref{eq:contradictionexpectation},
\begin{align}
\frac{1}{2} 
 &< \mathrm{Pr}_{x'\sim W(P')} \left[\min_{w \in \RR} \df(x',w \cdot b') < \frac{2 \Delta'}{\sum_{y\in P}\|y\|_{\infty}} \right] \nonumber \\
&\leq \mathrm{Pr}_{x'\sim W(P')} \left[{\min_{w \in \{-1,1\}} \df(x',w \cdot b')} < \frac{20 \Delta'}{\sum_{y\in P}\|y\|_{\infty}} \right] \label{eq:probapprox1}\\
 &=\mathrm{Pr}_{x'\sim W(P')} \left[x'\in C \right], \label{eq:probapprox2} 
\end{align}
where in \eqref{eq:probapprox1} we used \Cref{lem: approx_weights}. Now let $w_x^{\ast} \in \arg \min_{w\in\{-1,1\}} \df(x',w \cdot  b^{\ast})$ and let $w_x^{\prime} \in \arg \min_{w\in\{-1,1\}} \df(x',w \cdot  b^{\prime})$. 
Then, for all $x'\in C$, by the triangle inequality, our assumption \eqref{eq:contradictionassumption}, and since $\forall x,y\in\RR^d: \df(x,y)=\df(-x,-y)$, we have that 
\begin{align*}
 \df(x',w_x^* \cdot b^*) &\geq  \df(w_x' \cdot b',w_x^* \cdot b^*) - \df(x',w_x' \cdot b')\\
&> \frac{40 \Delta'}{\sum_{y\in P} \|y \|_{\infty}} - \frac{20 \Delta'}{\sum_{y\in P}\|y \|_{\infty}}\\
&= \frac{20 \Delta'}{\sum_{y\in P}\|y \|_{\infty}}, 
\end{align*}
which, by \eqref{eq:probapprox2} implies
\begin{align}\mathrm{Pr}_{x'\sim W(P')}\left[  \df(x', w_x^* \cdot b^*) > \frac{20 \Delta'}{\sum_{y\in P} \|y\|_{\infty}} \right] > \frac{1}{2}.\label{eq:beforecontradiction}\end{align} 
On the other hand, by Markov's inequality,  \eqref{eq:contradictionexpectation}, and \Cref{lem: approx_weights}, we have the following:
\begin{align*}
\frac{1}{2} &< 
\mathrm{Pr}_{x'\sim W(P')}\left[ \min_{w\in \RR}\df(x',w\cdot b^{\ast}) <\frac{2\Delta^{\ast}}{\sum_{y\in P} \|y\|_{\infty}}\right]\\
&\leq 
\mathrm{Pr}_{x'\sim W(P')}\left[ {\df(x', w_x^* \cdot b^*)} <\frac{20\Delta^{\ast}}{\sum_{y\in P} \|y\|_{\infty}}\right],
\end{align*} which contradicts \eqref{eq:beforecontradiction}. Therefore, $\min_{\substack{w\in \{-1,1\} } }\df(b',   w\cdot b^*) \leq \frac{40 \Delta'}{\sum_{x\in P} \|x\|_{\infty}}$. 
\end{proof}

\newpage
\section{Missing Proofs of Section \ref{sec:exact}}
\label{app:exact}

To show correctness of Algorithm \ref{alg: projection} we will use the following observations, which follow directly from the definition of traversal matrices.
\begin{obs}\label{obs: shrink traversal matrix}
        For a traversal matrix $M$ of dimension $(t,r)$  there exists a traversal matrix $M'$ of dimensions $(t-1,r-h)$ with $h \in \{0,1\}$ such that $M'(i,j) = M(i,j)$ for $1\leq i \leq t-1$ and $1\leq j \leq r-h$.
    \end{obs}
\begin{obs}\label{obs: extend traversal matrix}
        Given traversal matrix $M$ of dimensions $(t,r)$ and $h \in \{0,1\}$ then $M'\in\{0,1\}^{(t+1)\times(r+h)}$ defined as

       \[ M'(i,j)=\left\{\begin{array}{lr}
        M(i,j),  &\text{for } 1\leq i \leq t, 1\leq j \leq r\\
        1, &\text{for } i = t+1,j = r+h\\
       0, &\text{else}
        \end{array}\right \}\]
        is a traversal matrix.
\end{obs}

To simplify the notation in the proofs of this section, we introduce the following variation of the span, which fixes the weights. If the weights are not specified, we assume the unit weight setting.
 
\begin{definition} Let $Y=\{y_1,\dots, y_k\} \subset \RR^l$, and let $w=(w_1,\dots, w_k) \in \RR^k$ the Fréchet fspan of $Y$ is defined as
\[
\fspann(Y,w) = \bigcup_{t \in \NN } \left\{ \sum_{i=1}^k w_i x_i ~|~ x_i 
\in \closure_t(y_i) \right\}
\]
We may also write $\fspann(Y)$ where $w$ is defined to be the vector of all 1s. 
\end{definition}

\begin{observation}
The discrete Fr\'echet distance between $x$ and $\fspann(Y)$ can be written as 
\[
d_F(x,\fspann(Y)) = \min_{M_0,M_1,\dots, M_k} \| M_0 x - \sum_{i=1}^k M_i y_i\|_{\infty},
\]
where the minimum is over all traversal matrices $M_0,M_1,\dots, M_k$ of matching dimensions.
\end{observation}

The following two lemmas show correctness and running time of \Cref{alg: projection} leading to the proof of Theorem~\ref{lem: compute traversal matrices}.

\begin{restatable}{lemma}{lemmadynamicupdate}\label{lem: DF to span recursion}
For all $\vec{i} =(i_0,\ldots,i_k) \in [m]\times [l]^k$ it holds that \[D[\vec{i}] = d_F(x^{(i_0)}, \fspann (\{b_1^{(i_1)}, \dots, b_k^{(i_k)}\})).\]
\end{restatable}

\begin{proof}
     For  $\vec{i},\vec{j} \in [m] \times [l]^k$ we write $\vec{j} \prec_l \vec{i}$ if $\vec{j}$ is a predecessor of $\vec{i}$ in the lexicographical ordering of $[m]\times [l]^k$.
    For $\vec{i} \in [m]\times [l]^k$ and $0\leq r\leq k$ let $M^{(\vec{i})}_r$ be traversal matrices that satisfy 
    \[d_F(x^{(i_0)}, \fspann( \{b_1^{(i_1)}, \dots, b_k^{(i_k)}\})) = \Vert M^{(\vec{i})}_0 x^{(i_0)} - \sum^k_{r=1} M^{(\vec{i})}_r b_r^{(i_r)}\Vert_\infty.\] 

    We now show the statement by induction over $\vec{i} \in [m]\times[l]^k$. 
   
    Note that for $\vec{i} = (1,\ldots,1)$ it trivially holds 
    \[D[\vec{i}] = |x_1 - \sum^k_{r=1} (b_r)_1| = d_F(x^{(1)},\fspann(\{b^{(1)}_1,\ldots,b^{(1)}_k\})),\] 
    since the time series all have complexity $1$.
    
    Consider some arbitrary $\vec{i} \in [m]\times [l]^k$ with $(1,\ldots,1) \prec_l \vec{i}$ such that for $\vec{j} \in [m]\times [l]^k$ with $ \vec{j}  \prec_l \vec{i}$ it holds
    \[D[\vec{j}] = d_F(x^{(j_0)}, \fspann(\{b_1^{(j_1)},\ldots,b_k^{(j_k)}\})).\]
    Note that for all $h \in H_{\vec{i}}$ we have $\vec{i}-h\prec_l \vec{i}$ and therefore
     \[D[\vec{i}-h] = 
    d_F(x^{(i_0-h_0)}, \fspann(\{b_1^{(i_1-h_1)}, \ldots, b_k^{(i_k-h_k)}\})).\]
    By \Cref{obs: shrink traversal matrix}, there exists $h^*  \in H_{\vec{i}}$ and traversal matrices $M'_r$ for $0\leq r \leq k$ that satisfy
   
\[\Vert M^{(\vec{i})}_0 x^{(i_0)} - \sum^k_{r=1} M^{(\vec{i})}_r b_r^{(i_r)}\Vert_\infty = \max \big\{\Vert M'_0 x^{(i_0-h^*_0)} - \sum^k_{r=1} M'_r b_r^{(i_r-h^*_r)}\Vert_\infty , |x_{i_0}-\sum^k_{r=1} (b_r)_{i_r}|\big\}.\]
We therefore get
\begin{align*}
    D[\vec{i}] &= \max \big \{ \min_{h \in H_{\vec{i}}} D[\vec{i}-h ], |x_{i_0}-\sum^k_{r=1}(b_r)_{i_j}| \big\}\\
    &\leq \max \big \{ D[\vec{i}-h^* ], |x_{i_0}-\sum^k_{r=1}(b_r)_{i_r}| \big\}\\
    &= \max \big \{ d_F(x^{(i_0-h^*_0)}, \fspann(\{b_1^{(i_1-h^*_1)}, \dots, b_k^{(i_k-h^*_k)}\})), |x_{i_0}-\sum^k_{r=1}(b_r)_{i_r}| \big\}\\
    &\leq  \max \big \{\Vert M'_0 x^{(i_0-h^*_0)} - \sum^k_{r=1} M'_r b_r^{(i_r-h^*_r)}\Vert_\infty, |x_{i_0}-\sum^k_{r=1}(b_r)_{i_r}| \big\}\\
    &= \Vert M^{(\vec{i})}_0 x^{(i_0)} - \sum^k_{r=1} M^{(\vec{i})}_r b_r^{(i_r)}\Vert_\infty\\
    &= d_F(x^{(i_0)}, \fspann( \{b_1^{(i_1)}, \dots, b_k^{(i_k)}\}))
\end{align*}
On the other hand let $h'$ minimize
$d_F(x^{(i_0-h_0)}, \fspann(\{b_1^{(i_1-h_1)}, \dots, b_k^{(i_k-h_k)}\}))$ 
over $h \in H_{\vec{i}}$. By \Cref{obs: extend traversal matrix}, there exist traversal matrices $M''_r$ for $0\leq r \leq k$ that satisfy 
\[\Vert M''_0 x^{(i_0)} - \sum^k_{r=1} M''_r b_r^{(i_r)}\Vert_\infty = \max \big\{\Vert M^{(\vec{i}-h')}_0 x^{(i_0-h'_0)} - \sum^k_{r=1} M^{(\vec{i}-h')}_r b_r^{(i_r-h'_r)}\Vert_\infty,
    |x_{i_0}-\sum^k_{r=1} (b_r)_{i_r}|\big\}.\]
Therefore,  
we get 
\begin{flalign*}
    D[\vec{i}] &= \max \big \{\min_{h\in H_{\vec{i}}}  D[\vec{i}-h], |x_{i_0}-\sum^k_{r=1}(b_r)_{i_r}| \big\} &&\\
    &=\max \big \{ \min_{h\in H_{\vec{i}}} d_F(x^{(i_0-h_0)}, \fspann(\{b_1^{(i_1-h_1)}, \ldots, b_k^{(i_k-h_k)}\})), |x_{i_0}-\sum^k_{r=1}(b_r)_{i_r}| \big\}&&\\
     &= \max \big \{ d_F(x^{(i_0-h'_0)}, \fspann(\{b_1^{(i_1-h'_1)}, \ldots, b_k^{(i_k-h'_k)}\})), |x_{i_0}-\sum^k_{r=1}(b_r)_{i_r}| \big\}&&\\
     &= \max \big\{\Vert M^{(\vec{i}-h')}_0 x^{(i_0-h_0')} - \sum^k_{r=1} M^{(\vec{i}-h')}_r b_r^{(i_r-h'_r)}\Vert_\infty ,|x_{i_0}-\sum^k_{r=1} (b_r)_{i_r}|\big\}&&\\
     &= \Vert M''_0 x^{(i_0)} - \sum^k_{r=1} M''_r b_r^{(i_r)}\Vert_\infty&&
\end{flalign*}

And since $d_F(x^{(i_0)}, \fspann( \{b_1^{(i_1)}, \dots, b_k^{(i_k)}\}) \leq \Vert M''_0 x^{(i_0)} - \sum^k_{r=1} M''_r b_r^{(i_r)}\Vert_\infty$, we have  \[d_F(x^{(i_0)}, \fspann( \{b_1^{(i_1)}, \ldots, b_k^{(i_k)}\}) \leq D[\vec{i}],\]  which concludes the proof.

\end{proof}

\begin{restatable}{lemma}{fixedweightsdp}\label{lem: unweighted projection error}
Given time series $x \in \RR^m$ and $b_1,\ldots, b_k\in \RR^l$, Algorithm \ref{alg: projection} computes in $O(2^{k}  \cdot m \cdot l^{k})$ time a multidimensional array $D\subset \RR^{m\times l^k}$ such that for all $\vec{i} \in [m]\times[l]^k$ it holds  \[D[\vec{i}] = d_F( x^{(i_0)} ,\fspann(\{b^{(i_1)}_1,\dots, b^{(i_k)}_k\})).\]
\end{restatable}

\begin{proof}
The correctness of Algorithm \ref{alg: projection} follows from Lemma $\ref{lem: DF to span recursion}$. As for the running time it takes $O(m \cdot l^k)$ time to initialize $D$ and additionally, there are $O(m \cdot l^k)$ iterations, and each iteration takes $O(2^k)$ time.
\end{proof}

\end{document}